%% file: main.tex
\documentclass[openacc]{rstransa}
\usepackage[T1]{fontenc}
\usepackage[utf8]{inputenc}
\usepackage{physics}
\usepackage{hyperref}
\usepackage{enumitem}
\usepackage{amsfonts}
\usepackage{orcidlink}
\usepackage{fontawesome5}

\input{definitions.tex}
\usepackage{varioref}
\usepackage{cleveref}
\usepackage{array}
\usepackage{graphicx}
\usepackage{amsmath}
\usepackage{amssymb}

\usepackage{booktabs}
\usepackage{caption}
\usepackage{cancel}
\usepackage{esvect}
\usepackage[table]{xcolor} \usepackage{float}
\usepackage{soul}
\usepackage{tabularx, array, makecell, multirow, adjustbox, colortbl, rotating}
\newcolumntype{C}[1]{>{\centering\arraybackslash}m{#1}}
\newcolumntype{Y}{>{\centering\arraybackslash}X} 
\definecolor{tablegray}{HTML}{F8F8F8} \definecolor{tablegreen}{HTML}{FAFFFA}
\definecolor{tablered}{HTML}{FFFBF8}
\definecolor{tableyellow}{HTML}{FFFFF8}
\definecolor{egreen}{HTML}{009977}
\definecolor{eblue}{HTML}{0099FF}
\definecolor{ered}{HTML}{AA4400}

\newtheorem{theorem}{\bf Theorem}[section]

\newtheorem{corollary}{\bf Corollary}[section]
\newtheorem{definition}{\bf Definition}[section]

\newcommand{\raisedtarget}[1]{  \raisebox{\fontcharht\font`P}[0pt][0pt]{\hypertarget{#1}{}}}

\usepackage{float}

\renewcommand{\vec}[1]{\vv{#1}}

\titlehead{Research}

\begin{document}

\title{Clean up your Mesh! \\ \large{Part 1: Plane and simplex}}

\author{Steven De Keninck$^{1}$, Martin Roelfs$^{2}$, Leo Dorst$^{3}$ and David Eelbode$^{4}$}

\address{$^{1,3}$University of Amsterdam\\
$^{2,4}$University of Antwerp}

\subject{Discrete Differential Geometry, Mesh Processing}


\corres{Steven De Keninck\\
\email{s.a.j.dekeninck@uva.nl}}

\jname{rsta}

\begin{abstract}
We revisit the geometric foundations of mesh representation through the lens of Plane-based Geometric Algebra (PGA), questioning its efficiency and expressiveness for discrete geometry.  
We find how $k$-simplices (vertices, edges, faces, ...) and $k$-complexes (point clouds, line complexes, meshes, ...) can be written compactly as joins of vertices and their sums, respectively.
We show how a single formula for their $k$-magnitudes (amount, length, area, ...) follows naturally from PGA's Euclidean and Ideal norms.
This idea is then extended to produce unified coordinate-free formulas for classical results such as volume, centre of mass, and moments of inertia for simplices and complexes of arbitrary dimensionality.
Finally we demonstrate the practical use of these ideas on some real-world examples.
\end{abstract}

\begin{fmtext}
\includegraphics[width=0.6\textwidth]{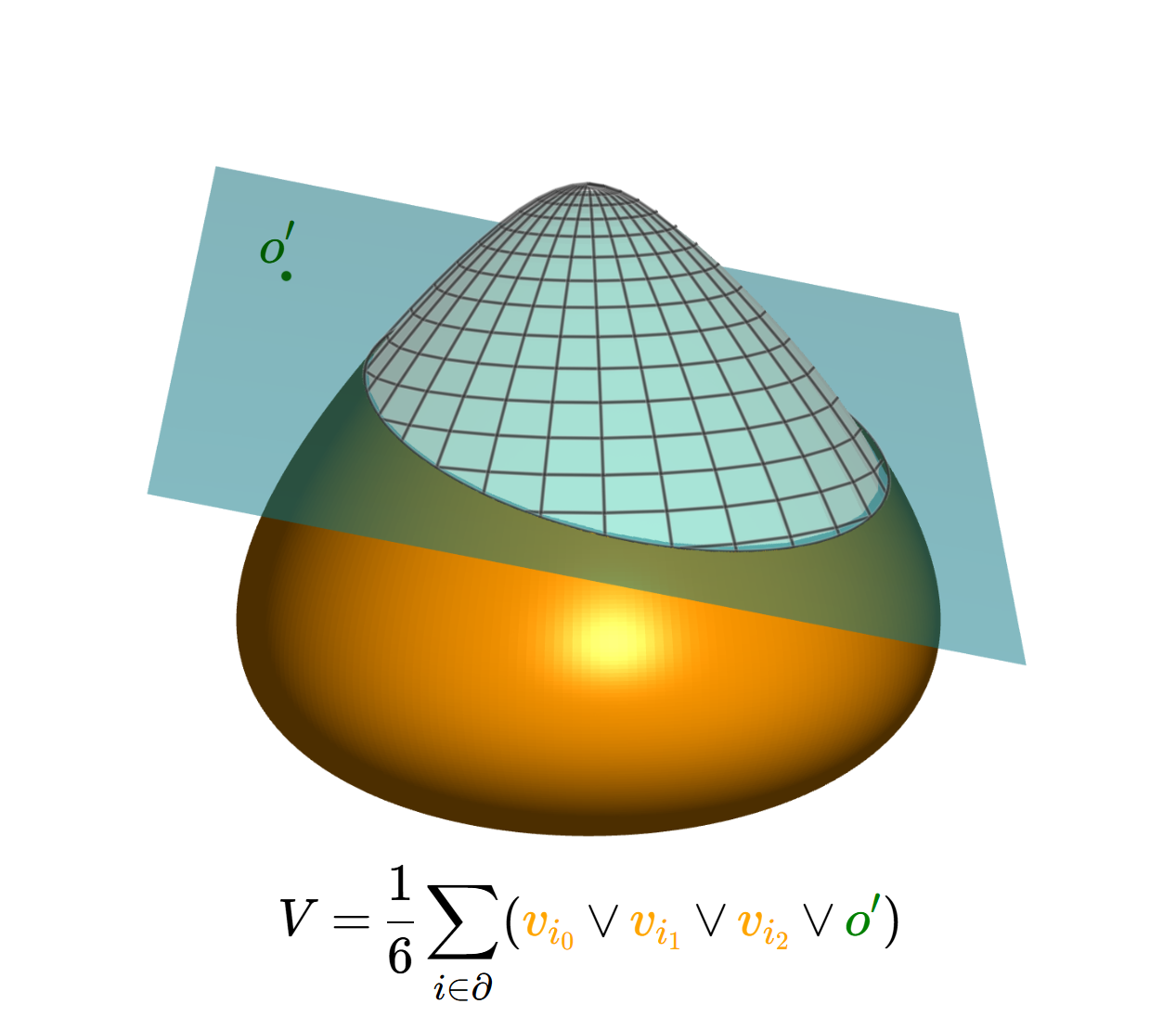}

\end{fmtext}
\maketitle

\section{Introduction}
In the geometric representation of discrete 
3D meshes, 
it sometimes seems that there are as many 3D mesh representations as there are engineering
problems involving them. For each application, a variety of topological
(e.g. half-edge, adjacency matrix, triangle lists, ...), representational
(e.g. position vectors, plane equations, edge vectors, ...), and performance
(e.g. octrees, bsp, bvh, ...) options need to be considered.
In this paper we  focus on the geometric representation, and consider 
the efficacy of the multivectors of PGA, plane-based geometric algebra, for the representation of simplices and complexes. 
We believe that this encompasses most of the others, in a compact, coordinate-free framework.
While we refer the reader to \cite{Dorst2007GeometricScience,Gunn2019GeometricGraphics,DeKeninck2024Normalization6D} for 
a more complete introduction to PGA, we shall briefly 
review the basics used in this article. 

The main contributions of this paper are formulas for the 0th (amount, length, area, ...),
1st (center of mass) and 2nd moments (inertia) of arbitrary $k$-simplices and complexes.
An overview of these results is given in \vref{table:cheat},
where the header compactly summarizes the main results to be proven in this work, while the rows present specific formulas that can be implemented directly using any modern GA library that supports PGA, e.g. \cite{gamphetamine,kingdon,python_clifford,ganja}.
The methodology distinguishes itself from previous methods using exterior calculus \cite{Crane:2013:DGP} or GA \cite{GA4Ph} through its inclusion of ideal elements (residing at infinity, but finitely represented in this projective algebra), which endows PGA with a Euclidean and ideal norm, neatly corresponding to magnitudes of simplices and their boundaries.
Moreover, these quantities are naturally invariant under the transformations of the Euclidean group \Eu{n}.

This paper is concerned with motivating and proving these results, with important geometric intuitions being given throughout. 
We show an involved practical example of measuring the volume of an airplane tank mesh, where PGA provides a very compact solution. 
\begin{table}[b]
\footnotesize
\fontfamily{cmss}\selectfont
\renewcommand{\arraystretch}{1.35}
\setlength{\tabcolsep}{0pt}
\centering
\rowcolors{1}{white}{tablegray} 
\begin{tabular}{|@{}p{0.5\textwidth}p{0.5\textwidth}@{}|}
\Xhline{0.2pt}
\multicolumn{2}{c}{
\begin{tabular}{|>{\centering\arraybackslash}p{0.24\textwidth}>{\centering\arraybackslash}p{0.24\textwidth}>{\centering\arraybackslash}p{0.24\textwidth}>
{\centering\arraybackslash}p{0.28\textwidth}|}
 $k$-simplex & 
 $k$-complex &
 $k$-magnitude & $k$-com\\
 $S_k = v_0 \vee ... \vee v_k$ & 
 {$C_k = \sum S_k$} & 
 {$\frac {1} {k!} \lVert S_k \rVert \, , \, \frac {1} {k!} \lVert C_{k-1} \rVert_\infty$} & 
 {$ \frac{ \sum_{\partial}(\sum v_i + o)(S_{k-1}\vee o)}{(k+1)!}$} \\
 \tiny point, edge, triangle, tetrahedron, ... & \tiny point set, polygon, mesh, ... & \tiny amount, length, area, volume, ... & \tiny center of mass\\
\end{tabular}
}\\
\arrayrulecolor{tablegray}
\Xhline{0.02pt}
\,\,Vertex $v$                       & $(x \mathbf{e}_1 + y \mathbf{e}_2 + z \mathbf{e}_3 + \mathbf{e}_0)^*$ \\
\Xhline{0.02pt}
\,\,Edge $\edg$                   & $v_1 \vee v_2$ \\
\Xhline{0.02pt}
\,\,Triangle $F$                     & $v_1 \vee v_2 \vee v_3$ \\
\Xhline{0.02pt}
\,\,Tetrahedron $T$                  & $v_0 \vee v_1 \vee v_2 \vee v_3$ \\
\Xhline{0.02pt}
\Xhline{0.02pt}
\,\,Vertex count                   & $\lVert \sum v_i \rVert$ \\
\Xhline{0.02pt}
\,\,Edge length                   & $\lVert \edg \rVert$ \\
\Xhline{0.02pt}
\,\,Face area                        & $\frac{1}{2}\lVert F \rVert$ \\
\Xhline{0.02pt}
\,\,Tetrahedron volume               & $\frac{1}{6}\lVert T \rVert$ \\
\Xhline{0.02pt}
\,\,Planar polygon area (via $\partial\edg$) & $\frac{1}{2}\lVert \sum \edg_i \rVert_\infty = \frac {1} {2} \sum (E_i \vee \origin) $ \\
\Xhline{0.02pt}
\,\,Mesh area (via $\partial F$)     & $\frac{1}{2}\sum \lVert F_i \rVert$ \\
\Xhline{0.02pt}
\,\,Mesh volume (via $\partial F$)   & $\frac{1}{6}\lVert \sum F_i \rVert_\infty = \frac {1} {6} \sum (F_i \vee \origin)$ \\
\Xhline{0.02pt}
\,\,Area of sum of missing {$\partial F$} & $\frac{1}{2}\lVert \sum F_i \rVert$ \\
\Xhline{0.02pt}
\,\,Mesh center of mass (via $\partial F$) & $\tfrac{1}{24}\sum (v_{i_1}+v_{i_2}+v_{i_3}+\origin)(F_i \vee \origin)$ \\
\arrayrulecolor{black}
\Xhline{0.2pt}
\end{tabular}
\vspace{3pt}
\caption{Magnitudes of simplices and complexes.  }
\vspace{3pt}
\label{table:cheat}
\end{table}
\noindent
\Cref{table:cheat} shows that the magnitudes of complexes of various dimensions follow a straightforward pattern when expressed in PGA.

\section{The Geometric View of Clifford Algebras}

This paper shows a use of 3D PGA $\R_{3,0,1}$, Plane-based Geometric Algebra, which is a novel model of Euclidean geometry \cite{Gunn2019GeometricGraphics}.
As a geometric algebra, it is an interpretation of Clifford algebra, which mathematically might be described as a tensor algebra modulo the ideal $x \otimes x - Q(x)$, with $Q(\gp)$ the chosen quadratic form (this effectively states that any square of a vector $x$ can be replaced by  the scalar $Q(x)$). However, Clifford algebra's  geometric interpretation is more directly explained by introducing an equivalent `geometric algebra' (GA) through elementary symmetries.

In the practice of GA, given the desire to work in a geometry (such as the Euclidean geometry of this paper), we set up a vector space with a quadratic form (metric) such that the elementary symmetry operations are representable as orthogonal transformations of the representational space. This is advantageous since GA has an effective representation of orthogonal transformations, by spinor-like elements called $k$-reflections. Let us show briefly how this geometric view indeed produces a practical Clifford algebra.

\subsection{Transformations as $k$-reflections}
According to the Cartan-Dieudonn\'e theorem, all orthogonal transformations in a $d$-dimensional space  $W$ with quadratic form $Q$ can be written as at most $d$ reflections in hyperplanes of $W$. So let us consider a hyperplane reflection of a vector $x \mapsto M(x)$, in a hyperplane through the origin with (non-null) normal vector $a$. From linear algebra, we easily derive the expression 
\[
M(x) = x - 2 (x \cdot a) a /(a \cdot a), 
\] 
where we introduced the dot product notation $a \cdot a \equiv Q(a)$ (which by polarization gives $a \cdot b = \half (Q(a+b) - Q(a) - Q(b))$).

In geometric algebra, we consider the dot product of vectors as the symmetric part of the {\em geometric product} (denoted by a half-space), $a \cdot b = \half (a \, b + b \, a)$. When we substitute this into $M(x)$, we get a more compact expression if we assume this product to be bilinear, distributive, associative (but not commutative), commuting with scalars:\footnote{... which makes this geometric product in fact identical to the Clifford product from the formal tensor-based Clifford algebra definition.}
\begin{equation}
	M(x) 
	= x - 2 (x \cdot a) \gp a / (a\cdot a) 
	= x - (x \gp a + a \gp x) \gp a / a^2 
	= x - x  - a \gp x \gp a / a^2 
	= - a \gp x \gp \inv{a}.
\end{equation}
In this final form, the multiple reflections of Cartan-Dieudonn\'e are easily performed. Two consecutive reflections in planes $a$ and $b$ lead to the $x  \mapsto b \gp a \gp x \gp \inv{a} \gp \inv{b} = (ab) \gp x \gp \inv{(ab)}$. Therefore the element $(ab)$ of the geometric algebra denotes 
a double reflection -- in the 3D algebra $\R_3$, this would be a rotation around the origin (and in fact isomorphic to a quaternion).

For a $d$-dimensional space $W$, its orthogonal transformations can thus be represented as elements consisting of the geometric product of $k$ (which is at most $d$) invertible vectors, called a {\em $k$-reflection}. Such a $k$-reflection $V$ is to be applied to a vector $x$ by the `sandwiching' product $x \mapsto  (-1)^k V\gp x \gp \inv{V}$. We will extend this to general elements below.

\subsection{Choosing a Geometric Algebra}
By a sensible choice of representational space, we can represent a variety of geometries. For the Euclidean geometry of this paper, we need to represent rotations and translations (and their combinations as screw motions) as multiple reflections. Since these operators can be made as multiple planar reflections, the GA approach dictates that we choose a representational space whose vectors represent planes in Euclidean space. 

The familiar homogeneous plane equation $a x + by + cz + d = 0$ involves the homogeneous 4-D `plane-vector' $[a,b,c,d]^T$. The space of such vectors is the representational space of 3D PGA, plane-based geometric algebra.

We need to establish a metric so that the multiple reflections correctly represent the Euclidean transformations, and it turns out that this is $\GR{3,0,1}$, 
This signature $(3,0,1)$ implies that we can make an orthogonal basis $\{\e{0},\e{1}, \e{2}, \e{3}\}$ such that:
\begin{align*}
\e{i} \cdot \e{i} = 1 &\mbox{~~~~for $i=1,2,3$} \\
\e{i} \cdot \eo = 0 & \mbox{~~~~for $i=0,1,2,3$} \\
\e{i} \cdot \e{j} = 0 &\mbox{~~~~for $i,j=0,1,2,3$ and $i\neq j$} 
\end{align*}
\noindent
So $\e{1}, \e{2}, \e{3}$ form an orthogonal basis of three Euclidean basis vectors and $\e{0}$ is an orthogonal null vector for which $\eo^2=0$. 
Using this basis, we would represent the plane with Euclidean unit normal vector $\vec{n}$ at signed distance $\delta$ from the origin, which has the equation $\vec n \cdot \vec x - \delta =0$, as a multiple of the PGA vector $\vec n - \delta \eo$.

The elements obtained by multiplying such vectors can be used to algebraically represent Euclidean transformations, and with an even number of factors are isomorphic to the dual quaternions, viewed as multiple reflections in planes.

For instance, the 2-reflection in two parallel planes $p_1 = \vec n - \delta_1 \e{0}$ and $p_2 = \vec n - \delta_2 \e{0}$ (with $\vec n^2 = 1$) produces the PGA element
\[
p_2 \gp p_2 = (\vec n - \delta_2 {\eo})\gp (\vec n - \delta_1 {\eo}) = 1 - (\delta_2-\delta_1) {\eo} \vec n
\equiv 1 - \eo \vec t,
\]
where we defined $\vec t \equiv (\delta_2-\delta_1) \vec n$, the distance vector between the planes. Then applying the element $p_2p_1$ in a sandwich product to a plane $\vec m$ at the origin gives:
\[
(1 -  \eo \vec t) \gp \vec m \gp (1 + \eo \vec t) 
= \vec  m - 2(\vec m \cdot \vec t) \gp \eo,
\]
which is the vector representing the $\vec m$-plane translated over $2\vec t$. This is indeed what we would expect from a double reflection in parallel planes. The product between two non-parallel planes produces a rotation operator, over twice their relative angle and with their common line as its axis (see e.g.~\cite{GSG}). The product of four (or more) vectors then encodes a general 3D Euclidean motion.

\subsection{Derived Products and their Semantics}
For this geometric use of Clifford algebra, it is convenient to define some more `products' derived from the fundamental geometric product.
We have seen how the dot product of two vectors is the symmetric part of their geometric product. Their anti-symmetric part is the {\em wedge product} $\wedge$:
\[
	a \gp b = a \cdot b + a \wedge b = \half(a \gp b + b \gp a) + \half(a \gp b - b \gp a).
\]
This wedge product establishes a Grassmann subalgebra within the geometric algebra. It can be extended by antisymmetry to multiple factors, and then forms elements called $k$-blades. In PGA, the $k$-blade $a \wedge b$ of two plane vectors $a$ and $b$ is a 2-blade representing their intersection (though note that $a \wedge a =0$, it is a linearized form of the geometric intersection operation). It is often called the `{\em meet product}' when used in this manner. 
We can make geometrical points in 3D PGA as 3-blades, since they are the intersection of 3 planes. 
The element $\e{0123} \equiv \eo \wedge\e1 \wedge \e2 \wedge \e3$ is chosen as the {\em pseudoscalar} of $\R_{3,0,1}$; it anti-commutes with all vectors.

The pseudoscalar can be used to define a (Hodge) duality in PGA, which we will denote by ${}^*$ (the exact definition will follow later in \cref{sec split}\ref {subsection_dual}). It associates, in a linear and invertible manner, a $(4-k)$-blade with each $k$-blade. The dual of the 3-blade point from PGA then becomes a 1-vector, in a more recognizable form reminiscent of the usual homogeneous coordinate representation. With this Hodge duality, we can also define the {\em join product} (aka regressive product) :
\begin{equation}\label{joindef}
	{(A \vee B)}^* = {A}^* \wedge {B}^*.
\end{equation}
In 3D PGA, this product allows two 3-blade points to be joined to form a 2-blade line connecting them. We will show the explicit parametrization (and its connection to Pl\"ucker coordinates) later in \cref{sec split}\ref{subsection_wedge},\ref{subsection_vee}.

The dot product can also be extended to work on multivectors, in a manner that obeys its quantitative duality with the wedge product. This is a bit involved (see e.g. \cite{Dorst2007GeometricScience}),  and since in PGA it is more natural to use the join instead, we will forego the details.

Ultimately, an element we can produce using these derived products and their sums is a {\em multivector}, i.e., a weighted sum of geometric products of basis vectors in the $2^4$-dimensional basis of $\R_{3,0,1}$. 
Algebraically, any such multivector $X$ transforms in exactly the same way (i.e., equivariantly) under $2k$-reflections, namely as the derived product sometimes called {\em sandwiching}:
\[
X \mapsto  V \gp X \gp \inv{V}.
\]
This is easy to see: linearity of the geometric product allows us to look at the basis blades only. A term like $\e1 \e2$ would be sandwiched to $V (\e1 \e2)\inv{V} = V \e1 (\inv{V} V) \gp \e2 \inv{V} = ((-1)^{2k} {V} \e1 \inv{V}) ((-1)^{2k}{V} \e2 \inv{}V)$, which is indeed the geometric product of the transformed basis vectors. This argument extends trivially to any number of basis vector factors. 
In PGA, these $2k$-reflections (aka as `even' transformations) represent rotations , translations and their compositions, and they are all we need in this paper (no net reflections). \footnote{For general $k$-reflections, we can split $X = X_+ + X_-$ in its even grade parts $X_+$ and its odd grade part $X_-$. Then the transformation formula is: $X \mapsto V X_+\inv{V} + \hat{V} X_- \inv{V}$. Incidentally, our derived products only produce all even or all odd multivectors when applied successively to vectors. }

The equivariance property of the GA $k$-reflection representation of transformations makes it superior to the matrix representation of linear algebra,
and it has been used to great advantage (most recently in Machine Learning \cite{pmlr-v202-ruhe23a,flashclifford2025}).

\subsection{Coordinate-free Constructions}

In 3D PGA, a plane, line, point, reflection, rotation, translation, transflection, rotoreflection
and screw are thus each written as an appropriate multivector
\[
X 
\]
without the need for coefficients, indices (abstract or otherwise), or augmentation
to indicate co- or contravariant transformation rules. 
All sensible (i.e., equivariant)  geometrical operations between such elements can be specified and performed at this level.  
\begin{itemize}
\item An (all even or all odd) multivector $X$ transforms
under the action of a normalized product of $k$ vectors $V$, universally as $\pm V X \inv{V}$.
\item
The transformation
between any same grade types $A$ and $B$ is always $\sqrt{BA^{-1}}$. That $2k$-reflection transforms $A$ to $B$ when used in a sandwich product on $A$. 
\item 
Any $B$ can be projected onto any invertible $A$
to produce $(B\cdot A)A^{-1}$. 
\item 
The unique element at the intersection of blades $A$ and $B$ is always its {\em meet}
$A \wedge B$, while the {\em join} $A \vee B$ produces the unique element that contains both. (With the caveat that degeneracy may produce a zero result, see \cite{PGA4CS}.)
\item 
The {\em logarithm}
of any $2k$-reflection $\log V = \edg$  produces a bivector, and conversely the
{\em exponential} of a bivector $V = \exp(E)$ a $2k$-reflection. This ties directly into the Lie algebra of the motions.
In 3D PGA, the logarithm of a simple rotation is its axis; for a translation that axis is of the form $\eo \vec t$, and resides in the plane at infinity. The logarithm of a general 3D Euclidean motion is its screw bivector.
\end{itemize}
In this paper, we stay at the multivector level of specification as much as we can, but will occasionally derive the classical coordinate expression to show the correspondence.

\section{Euclidean Primitives and their Split}\label{sec split}
\Cref{table:basis} shows the geometric interpretation of the relevant multivectors of PGA: those of single grades ($k$-blades and $k$-vectors) and elements with only even or only odd grades. 

\begin{table}[ht]
    \footnotesize
  \fontfamily{cmss}\selectfont
  \renewcommand{\arraystretch}{1.3}
  \setlength{\tabcolsep}{1pt}
  \centering
  \begin{adjustbox}{center,trim={0 0 0 0}, frame}
  \begin{tabularx}{\textwidth}{|>{\centering\arraybackslash}m{.3cm}|Y|YYY|YYY|Y|YYYY|YYYY|}       \Xhline{0.4pt}
        &
    \multicolumn{1}{|c|}{\textbf{s}} &
    \multicolumn{6}{c|}{\textbf{bivector}} &
    \multicolumn{1}{c|}{\textbf{pss}} &
    \multicolumn{4}{|c|}{\textbf{vector}} &
    \multicolumn{4}{c|}{\textbf{trivector}} \\[0.2pt]
    \Xcline{2-17}{0.4pt}
        \multirow{3}{*}{\rotatebox{90}{\scriptsize elements}} &
    \multicolumn{1}{|>{\centering\arraybackslash}m{.5cm}|}{$\color{egreen}1$ } &
    $\color{ered}\mathbf{e}_{23}$ & $\color{ered}\mathbf{e}_{31}$ & $\color{ered}\mathbf{e}_{12}$ &
    $\color{eblue}\mathbf{e}_{01}$ & $\color{eblue}\mathbf{e}_{02}$ & $\color{eblue}\mathbf{e}_{03}$ &
    $\color{eblue}\mathbf{e}_{0123}$ &
    \multicolumn{1}{|c}{$\color{egreen}\mathbf{e}_1$} & $\color{egreen}\mathbf{e}_2$ & $\color{egreen}\mathbf{e}_3$ & $\color{eblue}\mathbf{e}_0$ &
    $\color{eblue}\mathbf{e}_{032}$ & $\color{eblue}\mathbf{e}_{013}$ &
    $\color{eblue}\mathbf{e}_{021}$ & $\color{ered}\mathbf{e}_{123}$ \\[0.2pt]
    \Xcline{2-17}{0.2pt}
        &
    \multicolumn{1}{|c|}{\cellcolor{white}} &
    \multicolumn{3}{c|}{\cellcolor{tablegreen} line$_o \quad \scriptstyle \mathfrak {so}(3)$} &
    \multicolumn{3}{c|}{\cellcolor{tablered} line$_\infty \quad \scriptstyle \mathfrak {t}(3)$} &
    \multicolumn{1}{c|}{\cellcolor{white}} &
    \multicolumn{4}{|c|}{\cellcolor{white}{plane}} &
    \multicolumn{4}{c|}{\cellcolor{white}{point / direction}} \\[0.2pt]
    \Xcline{3-8}{0.2pt}
    &
    \multicolumn{1}{|c|}{} &
    \multicolumn{6}{c|}{\cellcolor{tableyellow} line / screw \quad  $\scriptstyle \mathfrak {se}(3)$} &
    \multicolumn{1}{c|}{} &
    \multicolumn{4}{|c|}{$\scriptstyle a\mathbf e_1 + b\mathbf e_2 + c\mathbf e_3 + d\mathbf e_0$} &
    \multicolumn{4}{c|}{$\scriptstyle (x\mathbf e_1 + y\mathbf e_2 + z\mathbf e_3 + w\mathbf e_0)^*$} \\[0.2pt]
    \Xhline{0.4pt}
            \multirow{4}{*}{\rotatebox{90}{\scriptsize transformations}} &
    \multicolumn{4}{|c|}{\cellcolor{tablegreen} quaternion $\quad \scriptstyle SO(3)$}  &
    \multicolumn{3}{c|}{\cellcolor{tablered}$\,\,$translation $\, \scriptstyle{{T(3)}}\,\,$} &
    \multicolumn{1}{c|}{\cellcolor{white}} &
    \multicolumn{4}{|c|}{\cellcolor{white}\multirow{2}{*}{plane-reflection}} &
    \multicolumn{4}{c|}{\cellcolor{white}\multirow{2}{*}{point-reflection}} \\[0.2pt]
    \Xcline{2-9}{0.2pt}
    &
    \multicolumn{8}{|c|}{\cellcolor{tableyellow} dual quaternion $\quad \scriptstyle{{SE(3)}}$} &
    \multicolumn{4}{|c|}{\tiny $E(3)$} &
    \multicolumn{4}{c|}{\tiny $E(3)$} \\[0.4pt]
    \Xcline{2-17}{0.4pt}
    &
    \multicolumn{7}{|c|}{2-reflection} &
    \multicolumn{1}{c|}{} &
    \multicolumn{4}{|c|}{1-reflection} &
    \multicolumn{4}{c|}{}\\[0.2pt]
    \Xcline{2-17}{0.2pt}
    &
    \multicolumn{8}{|c|}{4-reflection\quad even multivectors\quad $\mathbb R_{3,0,1}^+$} &
    \multicolumn{8}{|c|}{3-reflection\quad odd multivectors\quad $\mathbb R_{3,0,1}^-$}\\[0.2pt]
        \Xhline{0.2pt}
      \end{tabularx}
  \end{adjustbox}
  \vspace{5pt}
  \caption{$\mathbb R_{3,0,1}$ basis, elements, transformations. Here "s" denotes scalar, "pss" denotes pseudoscalar. Subscript "o" means at the (arbitrary) origin, "$\infty$" at infinity (or ideal). }
  \label{table:basis}
\end{table}

\subsection{The Euclidean Split and Dualization}\label{subsection_dual}
PGA is thus a projective way of modeling Euclidean geometry, departing from its planes modelled as vectors. Its null basis vector $\e{0}$ can be interpreted as the {\em plane at infinity}, also know as the {\em ideal plane}. Rotations around an ideal line in this ideal plane model translations, which to us Euclidean beings feel rather different from regular rotations around a finite axis. Also, finite points feel different from directions, which in PGA are modelled as ideal points (i.e., points at infinity, on the ideal plane).

When we establish useful formulas, this distinction between the `finite' and the `ideal' is important, even though such elements may be united within a given PGA multivector to represent proper geometry. So let us introduce 
the {\em Euclidean split}, breaking up our multivectors into Euclidean and Ideal terms:
    \begin{equation}\label{split}
    A = A_E + \mathbf e_0 A_\ideal \ ,
    \end{equation}
    where the {\em Euclidean term} is $A_E$ and the Ideal term contains  $A_\ideal$, which we will refer to as the {\em Ideal factor}.
    Both $A_E$ and $A_I$ are elements of the subalgebra $\mathbb R_3$.
\noindent
In an implementation, it is easy to make this split: one just selects which parts are given on the $2^3$-dimensional Euclidean basis (involving only the vectors $\e{1}\e{2},\e{3}$ and their products), and which on the $2^3$ basis elements also having a factor $\e{0}$.

For the definition of the join, we need a PGA dualization operation.
The {\em Hodge dual} (aka right compliment dual) of any basis-blade $\e{J}$ with $J$ a multi-index is defined to satisfy $\e{J} \gp \e{J}^* = \e{0123}$.  By demanding linearity, this determines the Hodge dual of any multivector.

Using the Euclidean split, the Hodge (un)dual of a general multivector $A$ can be written using the {\em grade involution} (hat\footnote{$\widehat X$, the grade involution of $X$, inverts the sign of the odd grade parts. Note that $\widehat{A B} = \widehat{A} \widehat{B}$.}), the {\em reversion} (tilde\footnote{$\widetilde X$, the {\em reversion} of $X$, inverts the sign of grades 2 and 3 (modulo 4). Note that $\widetilde{A B} = \widetilde{B} \widetilde{A}$.}) and the Euclidean 3D
pseudo-scalar $I_3 = \mathbf e_{123}$ as
\begin{equation}\label{dual}
A^* = \widetilde{A_\ideal}\gp I_3 + \mathbf e_0\gp \widehat{\widetilde{A_E}} \gp I_3, \qquad A^{-*}  = - \widehat{\widetilde{A_\ideal}}\gp I_3 + \eo \gp \widetilde{A}_E \gp I_3 ,
\end{equation}
where $A^{-*}$ denotes the undual.
One may recognize the Euclidean duals of $A_E$ and $A_\ideal$ in these expressions.
We can trivially verify that $(A^{*})^{-*} = (A^{-*})^* = A$.
\begin{proof}
To demonstrate that $A^{*}$ of \cref{dual} is indeed the
right complement dual in PGA, it suffices to show that $ AA^{*} = A^{-*} A = \mathbf e_{0123}$, 
for all basis blades $A=\e{J}$ of $\R_{3,0,1}$.
This can be done by direct computation, when one remembers that for basis blades we have $A_E A_\ideal =0$, since they are either purely Euclidean or purely Ideal.
\end{proof}
We can now give the split definition of the PGA meet and join, reducing them to computations in the Euclidean subalgebra $\R_3$:
\begin{align}
    (A_E + \e0 \gp A_I) \wedge (B_E + \eo \gp B_I)
    &= (A_E \wedge B_E) + \eo \gp (\widehat{A_E} \wedge B_I + A_I \wedge B_E)
\end{align}
and 
\begin{align}
    \lefteqn{(A_E + \e0 \gp A_I) \vee (B_E + \eo \gp B_I)} \nonumber \\
                  &= 
   \big( ( \widehat{B_I} I_3) \wedge (A_E I_3) \big) \gp I_3^{-1}
   -\big( (B_E I_3) \wedge (A_I I_3) \big) \gp I_3^{-1}  
   - \e{0} \gp \big( ( B_I I_3) \wedge (A_I I_3)\big) \gp I_3^{-1}.
    \end{align}
The terms in the result are effectively joins in $\R_3$, see \cref{joindef} and also \cite{Dorst2007GeometricScience}.

\subsection{Norm and Ideal Norm.}\label{subsection_norms}

The Euclidean split \cref{split} leads us to define two useful PGA norms. The Euclidean norm of a PGA element $A$ is equal to the norm of its Euclidean term $A_E$ and  given simply by 
\[
\lVert A \rVert \equiv \sqrt{\tilde AA} = \lVert A_E \rVert.
\]
Because $\e0^2=0$, the $A_\ideal$ factor does not impact the Euclidean norm. We will however be interested in the magnitude of $A_\ideal$, so we define the ideal norm as the norm of that ideal factor, and give three equivalent alternatives to compute it:
\[
\lVert A \rVert_\infty \equiv 
\lVert A_\ideal \rVert  = \lVert   A  \vee \origin \rVert = \lVert A^* \rVert.
\]
Here $\origin$ is our notation for the point at the origin $\origin = \eo^* = I_3$. 
Using $\origin$, the ideal factor $A_\ideal$ may be extracted from $A$ as $A \vee \origin = \origin \vee \widehat{A}$ (since $A_E \vee \e{0}^* + (\e{0} A_\ideal) \vee \e{0}^* = A_\ideal$), and we will often rewrite it in this form.

Note that in Euclidean PGA, both these norms are
positive semi-definite and hence cannot be used to extract signed magnitudes.
An exception is when $ A \vee \origin $ is a scalar, this scalar then representing the signed ideal magnitude (this happens when $A$ is a $1$-vector, i.e., represents a hyperplane).  In that case, we use the signed expression $A \vee \origin$ instead of $\norm{A}_\infty$. 

We can normalize elements through dividing them by the appropriate norm. A general element with non-zero Euclidean term is usually normalized by the Euclidean norm, and we denote this normalization by an overline. So for a plane $p = \alpha (\vec n - \delta \eo)$ with ${\vec n}^2=1$, we have $\overline{p} = \vec n - \delta \eo$ (and note that then $\overline{p}^2 = {\vec n}^2 = 1$). A point $v$ at location $\vec v$ is normalized to the form $\overline{v} = I_3 - \eo \vec{v} I_3$, and ${\vec v}^2=-1$, as you would expect from a point reflection operator in 3D. We will mostly work with normalized points and vertices in this paper, so to unclutter the formulas we will drop the overline for them. 
If an element has only an ideal term, we would normalize it by the infinity norm.

\subsection{From Planes to Lines and Points}\label{subsection_wedge}

In this section we show how in PGA the outer product of planes generates lines and points as blades, and that the Euclidean split of those blades contains useful and recognizable information on relevant magnitudes, in preparation for their use in mesh processing.

As we have seen, a homogeneous linear equation representing a plane is embedded as a PGA vector:
\begin{equation}\label{plane}
ax + by + cz + d = 0 \stackrel{\text{embed}}{\implies} p = a\mathbf e_1 + b\mathbf e_2 + c\mathbf e_3 + d\mathbf e_0
\stackrel{\text{split}}{\implies} \vec n + \mathbf e_0 d
\end{equation}
where the Euclidean split in the last rewrite exposes the normal vector $\vec n = a\mathbf e_1 + b\mathbf e_2 + c\mathbf e_3$ and $-d$, the plane's signed distance
from the origin (scaled with the norm $\sqrt{a^2+b^2+c^2}$ of the normal vector). The outer product of two such planes in Euclidean split form,
\begin{equation}\label{wedge}
\begin{aligned}
p_1 \wedge p_2 = (\vec n_1 + \mathbf e_0 d_1) \wedge (\vec n_2 + \mathbf e_0 d_2) &= 
\vec n_1 \wedge \vec n_2 + \mathbf e_0(d_1\vec n_2 - d_2\vec n_1)\\\end{aligned}
\end{equation}
reveals, on a bivector basis, the well known cross product components $[ \vec n_1 \times \vec n_2\,;\, d_1\vec n_2 - d_2\vec n_1 ]$ for the Pl\"ucker coordinates of their intersection 
line. The norm and ideal norm for such a line
\begin{equation}\label{wedgenorm}
\begin{aligned}
\lVert p_1 \wedge p_2 \rVert &= \lVert \vec n_1 \wedge \vec n_2 \rVert = \lVert \vec n_1 \rVert \lVert \vec n_2 \rVert \sin \theta \\
\lVert p_1 \wedge p_2 \rVert_\infty &= \lVert d_1\vec n_2 - d_2\vec n_1\rVert\\
\end{aligned}
\end{equation}
encode the relevant angles and distances. For normalized planes (with $p_1^2 = p_2^2 =1$), the Euclidean norm 
produces the sine of the angle between the planes, while the ideal norm produces the 
distance between the planes for parallel planes, and the absolute distance from the line to 
the origin otherwise.

We can construct a point 3-vector $v$ at a location $\vec{v}$ by intersection the 3 coordinate planes $p_i = \e{i} - v_i\eo$ specifying the point coordinates $v_i$. This results in 
\begin{equation} \label{point}
v = \e{123} - \eo (v_1 \e{23} + v_2 \e{31} + v_3 \e{12})
= \e{123} - \eo \vec{v}\e{123}
= (\e0 + \vec{v})^*. ~~
\end{equation}
\noindent
The result clearly shows how the usual homogeneous point representation of the parametrized form $(1, \vec{v})$ embeds naturally into PGA through the Hodge dual.

In PGA, the numerical representation of a vertex does not depend on the choice of origin $\origin$. For if we choose as origin instead of  $\origin = \e{123} = I_3$ a point $q = I_3 -\eo \vec{q} I_3 $ at location $\vec{q}$, the vertex $v$ relative to that new location would have position vector $\vec{v}-\vec{q}$. Vertex $v$ would therefore be constructed as $v = (I_3 -\eo \vec{q} I_3) - \e0 (\vec{v} -\vec{q}) I_3 =
I_3 -\eo \vec{v} I_3$, i.e., it would be literally unchanged. Specifying vertices numerically in PGA notation therefore does {\em not} mean that we have chosen a fixed origin: the vertex's PGA `coordinates'  can be used relative to any origin. The corresponding relative position vector appears automatically as one performs such an `origin split' on the invariant vector $v$.

Using this representation of a point, it is easy to verify that $v \vee p =0$ for a plane $p = \vec n - \delta \eo$ indeed reproduces the normal plane equation $\vec v \cdot \vec n -\delta = 0$.

\subsection{From Points to Lines and Planes}\label{subsection_vee}

The PGA {\em join} of points generates lines and planes as blades. Again, the Euclidean split of those blades contains useful and recognizable information on relevant magnitudes, in which we recognize common expressions.
The involved coordinate expressions commonly used to compute those are therefore just retrievable aspects of the join of PGA points.

A normalized point $v$, at Euclidean position vector $\vec v = x\mathbf e_1 + y\mathbf e_2 + z\mathbf e_3$ 
is embedded as a PGA dual vector. i.e., a 3-blade
\begin{equation}
\begin{bmatrix}x \\ y \\ z\end{bmatrix} \stackrel{\text{embed}}{\implies} v =  (\eo + \vec v)^* \stackrel{\text{split}}{\implies} I_3 - \mathbf e_0\vec v I_3.
\end{equation}
The join between two points $v_0, v_1$ (assumed normalized unless otherwise mentioned) produces the carrier line connecting them:
\begin{equation}\label{join}
\begin{aligned}
v_0 \vee v_1 &= (\vec v_1 - \vec v_0)I_3 + \mathbf e_0{ (\vec v_1 \wedge \vec v_0) I_3}.
\end{aligned}
\end{equation}
This expression on a 6-dimensional 2-vector basis is directly related to the Pl\"ucker coordinates $(\vec v_1 - \vec v_0, \vec v_0 \times \vec v_1 )$, using the fact that the cross product between two vectors is Euclidean dual to their outer product: $\vec v_0 \times \vec v_1 = (\vec v_1 \wedge \vec v_0)I_3$. It shows direction vector and moment of the line.

\noindent
It immediately follows that the norm and ideal
norm of such a  carrier line
\begin{equation}\label{norms}
\begin{aligned}
\lVert v_0 \vee v_1 \rVert &= \lVert \vec v_1 - \vec v_0 \rVert,  \\
\lVert v_0 \vee v_1 \rVert_\infty &= \lVert \vec v_0 \wedge \vec v_1 \rVert = \lVert \vec v_0 \rVert \lVert \vec v_1 \rVert \sin \theta.
\end{aligned}
\end{equation}
denote respectively the distance between the joined points, and twice the area of the triangle
formed by our two points and the origin, since $\lVert v_0 \vee v_1 \rVert_\infty  = \lVert v_0 \vee v_1 \vee \origin\rVert$. The figure below illustrates all the relevant quantities.
\begin{center}
    \includegraphics[width=0.6\textwidth]{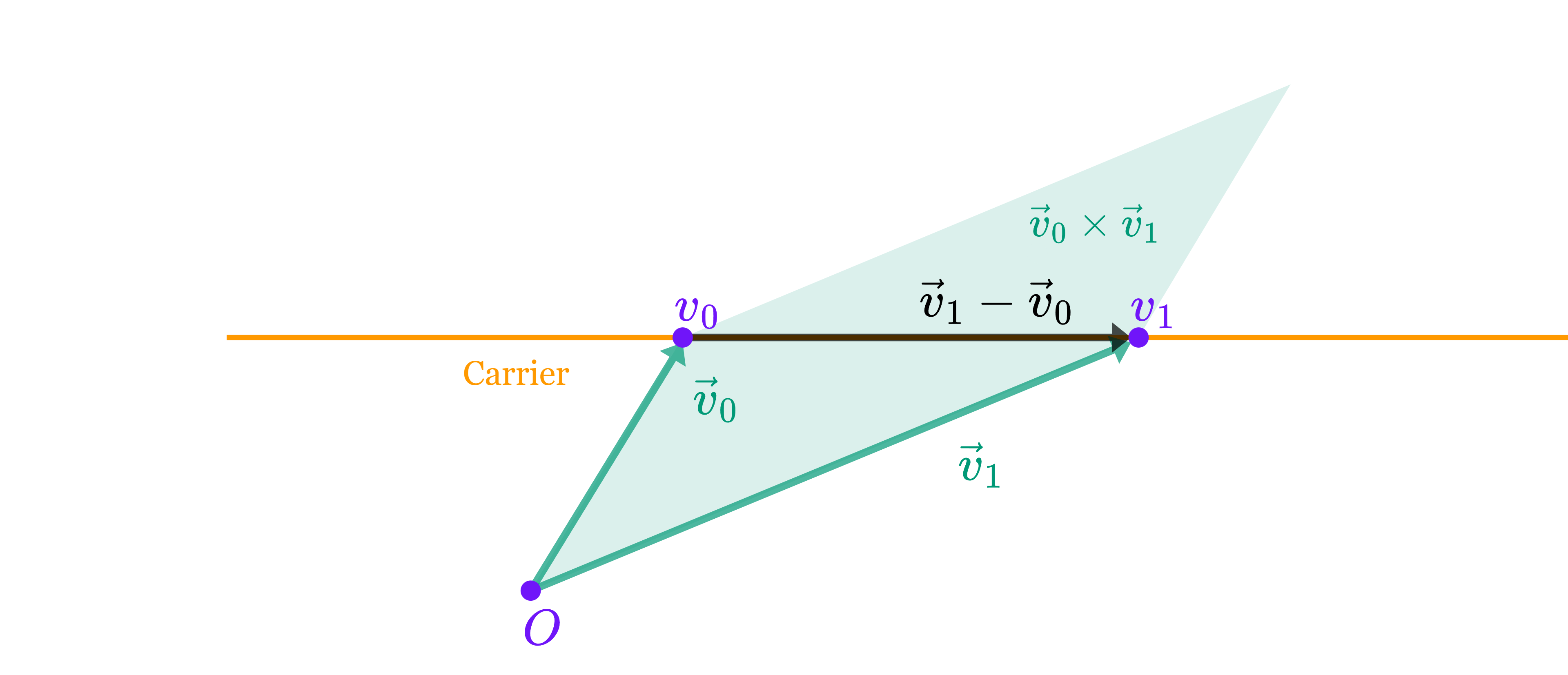}
\end{center}
\noindent
Let us join in a third point $v_2$, giving the quantitative representation of a 3D carrier plane:
\begin{equation}\label{three}
\begin{aligned}
v_0 \vee v_1 \vee v_2 
&= ( (\vec v_2 - \vec v_0) \wedge (\vec v_1 - \vec v_0) )I_3 + \mathbf e_0 (\vec v_0 \wedge \vec v_1 \wedge \vec v_2 ) I_3.
\end{aligned}
\end{equation}
As before, both norms contain valuable  metric information. We rewrite slightly:
\begin{equation}\label{norm3}
\begin{aligned}
\lVert v_0 \vee v_1 \vee v_2 \rVert 
&= \lVert (\vec v_2 - \vec v_0) \wedge (\vec v_1 - \vec v_0) \rVert , \\
&= \lVert \vec v_0 \wedge \vec v_2 + \vec v_2 \wedge \vec v_1 + \vec v_1 \wedge \vec v_0 \rVert \\
\lVert v_0 \vee v_1 \vee v_2 \rVert_\infty &= \lVert \vec v_0 \wedge \vec v_1 \wedge \vec v_2  \rVert = \lVert(\vec v_0 \cross \vec v_1) \cdot \vec v_2 \rVert \ . 
\end{aligned}
\end{equation}
For the Euclidean norm we see the sum of the bivectors (and their associated signed areas) for each (consistently ordered)
pair of vectors. This results in twice the area of the triangle defined by our points as can be seen in \vref{fig:simplicial_extension}.
As the ideal norm, we recognize the scalar triple product, known to produce the volume of the parallelepiped spanned
by our three vectors, or, more geometrically, 6 times the volume of the tetrahedron defined by our three points and the origin.

Joining with a fourth point $v_3$ gives the last carrier in the series, which is a scalar: 
\begin{equation}
v_0 \vee v_1 \vee v_2 \vee v_3  
= \big(\,(\vec v_3 -\vec v_0) \wedge 
(\vec v_2 -\vec v_0) \wedge 
(\vec v_1 -\vec v_0) \,\big)\gp I_3,
\end{equation}
which equals $3!=6$ times the oriented volume of the tetrahedron spanned by the points.
All these results are summarized in \cref{tab:la_vs_pga_norms}, which illustrates the clarity offered by PGA. The pattern even continues beyond 3D.

\begin{table}[!h]
\scriptsize
\renewcommand{\arraystretch}{1.6}
\setlength{\tabcolsep}{0pt}
\centering
\rowcolors{1}{white}{tablegray}
\begin{tabular}{|@{}>{\centering\arraybackslash}p{0.31\textwidth}
                |>{\centering\arraybackslash}p{0.42\textwidth}
                |>{\centering\arraybackslash}p{0.25\textwidth}@{}|}
\arrayrulecolor{tablegray}
\Xhline{0.2pt}
\centering\arraybackslash \textbf{magnitude} & \centering\arraybackslash \textbf{Vector Calculus} & \centering\arraybackslash \textbf{PGA} \\
\Xhline{0.02pt}
amount of $[v_0]$ & $?$ & $\norm{v_0}$ \\
length of edge 
$[v_0, v_1]$ & $\norm{\vec v_1 - \vec v_0}$ & $\norm{v_0 \vee v_1}$ \\
area of triangle $[v_0, v_1, v_2]$ & $\tfrac{1}{2} \norm{(\vec v_1 - \vec v_0) \cross (\vec v_2 - \vec v_0) }$ & $\tfrac{1}{2} \norm{v_0 \vee v_1 \vee v_2}$ \\
volume of tetrahedron $[v_0, \ldots, v_3]$ & $\tfrac{1}{3!} \norm{[(\vec v_1 - \vec v_0) \cross (\vec v_2 - \vec v_0)] \cdot (\vec v_3 - \vec v_0) }$ & $\tfrac{1}{3!} \norm{v_0 \vee v_1 \vee v_2 \vee v_3}$ \\
magnitude of $5$-cell $[v_0, \ldots, v_4]$ & $?$ & $\tfrac{1}{4!} \norm{v_0 \vee v_1 \vee v_2 \vee v_3 \vee v_4}$ \\
\arrayrulecolor{black}
\Xhline{0.2pt}
\end{tabular}
\vspace{3pt}
\caption{Various $k$-magnitudes as computed with vector calculus compared to the PGA norms.}
\label{tab:la_vs_pga_norms}
\end{table}

\section{Simplices and Complexes}\label{Meshes}
Not only the Clifford Algebra, but also the idea of a general $k$-simplex as a dimensional generalization of vertices, edges, triangles, tetrahedrons, et cetera was first introduced by W.K. Clifford \cite{clifford1866problem}, albeit not in the context of geometry but to solve a problem in probability. 
While he called it a ``prime confine'', in this paper we use the more modern name {\em $k$-simplex}.
Additionally we will use the term {\em $k$-complex} to describe a collection of $k$-simplices.
Moreover, a $k$-complex is called {\em complete} 
if for each of its $k$-simplices it also contains the constituent $(k-i)$-simplices for $1 \leq i \leq k$.
For example, for each triangle in a triangle mesh it also contains the edges and the vertices of that triangle.
Unless stated otherwise, we assume all complexes to be complete complexes for the remainder of this paper.
The boundary of a $k$-simplex is a $(k-1)$-complex, and the boundary of a boundary of a closed manifold vanishes.

\subsection{$k$-simplices, Carriers, Carrier Gauges and Simplicial Extension}\label{simplices}
A $k$-simplex $\sigma_k$ is a polytope that is the convex hull of its affinely independent $k+1$ vertices, i.e.
the simplest shape defined by a given ordered set of vertices $\sigma_k = [ v_0, \ldots, v_k ]$.\footnote{The square brackets employ the $k$-chain notation from algebraic topology \cite{munkres1984}.} In order of increasing $k$ these are a vertex, an edge, a triangle, and a tetrahedron.
\noindent\makebox[\textwidth][c]{\includegraphics[width=.7\linewidth]{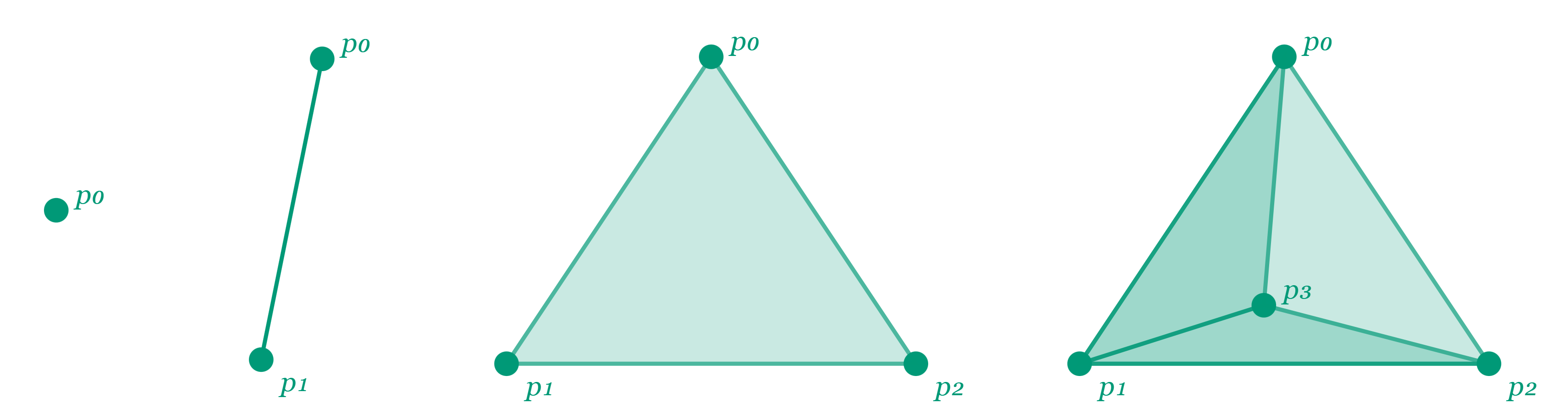}}
We define the $k$-\emph{carrier} of a $k$-simplex $\sigma_k$ as its $k$-dimensional infinite extension, with an appropriate $k$-magnitude assigned to it. 
The carrier of an edge is a line with an edge length measure, the carrier of a triangle is a plane with an area measure, etc. 
Moreover, we define the \emph{simplicial extension} of a $(k-1)$-simplex as the addition of a single vertex $v_k$, which we call the apex, creating a $k$-simplex.  
Let us begin with some geometric observations about $k$-magnitudes (length, area, volume, \ldots) of $k$-simplices.

\begin{minipage}{\linewidth}
  \centering
  \includegraphics[width=0.85\linewidth]{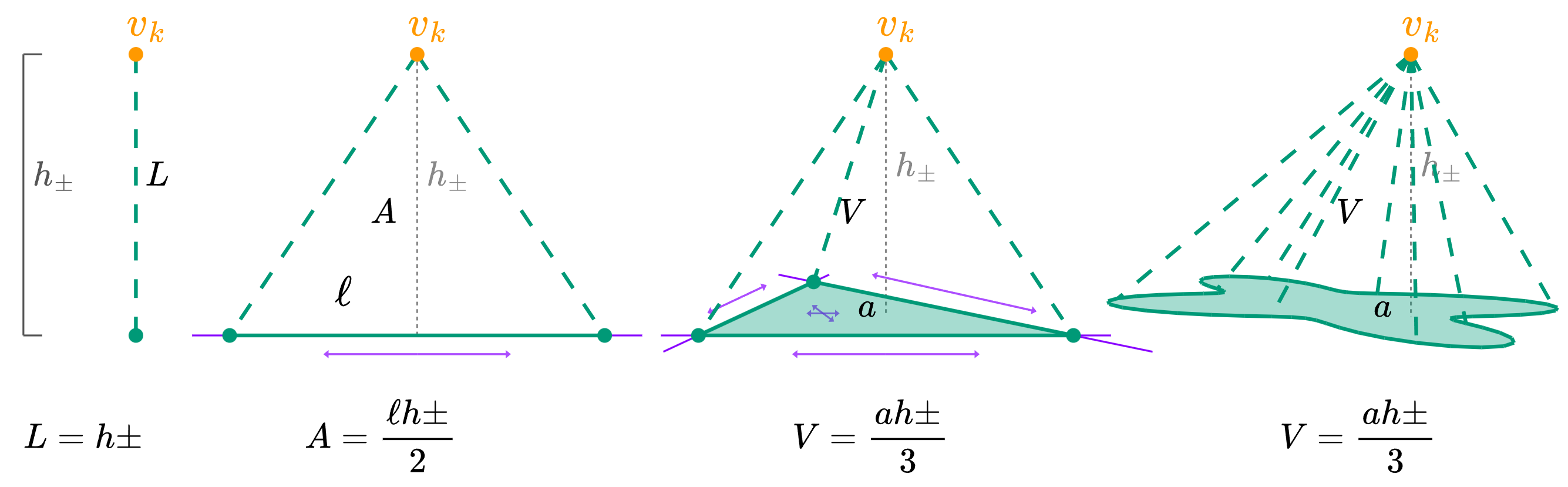}
  \captionsetup{type=figure,skip=2pt}   \captionof{figure}{\small The simplicial extension of a simplex and the associated (signed) $k$-magnitudes. 
  The carrier gauges of the $(k-1)$-carrier are shown as purple arrows.}
  \label{fig:simplicial_extension}
\end{minipage}
In \cref{fig:simplicial_extension}, note that the $k$-magnitude of each $k$-simplex is not affected by the  precise position of the apex $v_k$ (nor any of the other vertices).
Instead, it depends only on the $(k-1)$-magnitude of the $(k-1)$-carrier and the height of the apex vertex $v_k$ to the $(k-1)$-carrier.
For example, any edge of a triangle can be moved over its carrier line without changing the triangle area.
Similarly, each triangle of a tetrahedron can be moved on its carrier without changing the tetrahedron's volume. 
This degree of freedom, which we will call a \emph{carrier gauge}, must be present in any suitable formal definition of the carrier of a $k$-simplex.

As we have already seen how the join product of points creates spaces of higher dimensions that contain those points in \cref{sec split}\ref{subsection_vee}, we define the $k$-carrier of a $k$-simplex as follows:
\begin{definition}[$k$-carrier of a $k$-simplex ${\sigma_k = \left[v_0, \ldots, v_k\right]}$ ]\label{generaljoin} 
\begin{equation*}
{S(\sigma_k) = v_0 \vee \cdots \vee v_k}
\end{equation*}
When there is no ambiquity, we may write $S_k = S(\sigma_k)$.
\end{definition}
\noindent
By construction $S_k$ is the desired carrier space, but we must still verify that it 
(i) captures the correct magnitude, (ii) captures carrier gauges, and (iii) is oriented.

Firstly, in order to show that \cref{generaljoin} correctly measures the magnitude of a $k$-simplex, consider the $k$-simplex $\sigma_k = [v_0, \ldots, v_{k-1}, v_k]$, where we take $v_k$ to be the apex vertex. Using barycentric coordinates, $v_k$ can be expressed as
    \[ v_k = \sum_{i=0}^{k-1} \alpha_i v_i + h \vec u_k^* \ , \]
where $\sum_i \alpha_i = 1$ and $\vec u_k^*$ is a new affinely independent normalized infinite point. Since $v_i \vee v_i=0$,
    \[ S_k = v_0 \vee \cdots \vee v_{k-1} \vee v_k = S_{k-1} \vee (h \vec u_k^*) \]
and hence $\norm{S_k} = \norm{S_{k-1}} \norm{h}$. 

Secondly, the invariance of $S_k$ under transformations of the base $S_{k-1}$ follows from the equivariance of the join product. 
Specifically, let $V \in \SE{k-1}$ be a Euclidean motion in the carrier space $S_{k-1}$. Then $v_i' = V v_i V^{-1}$ are the transformed vertices within the carrier space $S_{k-1}$, but the equivariance allows us to write
    \[ S_{k-1}' = v_0'\vee \ldots \vee v_{k-1}' = V S_{k-1} V^{-1} = S_{k-1} \ , \]
and hence $S_k$ is invariant under transformations that occur in its subcarriers.

Thirdly, the carrier $S_k$ must share the orientation of the $k$-simplex. In other words, swapping two neighboring vertices in a simplex swaps the orientation, i.e. $[\ldots, v_i, v_j, \ldots] = -[\ldots, v_j, v_i, \ldots]$. 
This is captured faithfully by the anti-symmetry of the join product under swapped operands, i.e. $\cdots v_i \vee v_j \cdots = -\cdots v_j \vee v_i \cdots$, and hence a natural property of the carrier.

In conclusion, $S_k$ has the simple interpretation as the oriented carrier of the simplex $\sigma_k$, and its Euclidean norm represents the $k$-magnitude (amount, length, area, ...) of the parallelotope formed by $\sigma_k$, and it exhibits invariance under carrier gauges.
This motivates the following definition of the $k \geq 0$ dimensional magnitude (amount, length, area, volume, ... ) of a $k$-simplex $\sigma_k$ 
(vertex, edge, triangle, tetrahedron, ...), where we introduced a factor $1/k!$ to correct for the over-counting in $S_k$:
\begin{definition}[Magnitude of a $k$-simplex $\sigma_k$]
    \begin{equation*}\label{magnitude}
        \kmag \sigma_k \coloneqq \frac{1}{k!}\lVert S_k \rVert = \frac{1}{k!}{\lVert v_0 \vee \cdots \vee v_k \rVert} \ .
    \end{equation*}
\end{definition}
\noindent
The computation of the $k$-magnitude of a simplex $\sigma_k$ is therefore straightforward.
There is however, an alternative method to compute the magnitude of a $k$-simplex from its boundary, which leverages PGA's inclusion of elements at infinity.

\subsection{Mind the gap}\label{sec:mindthegap}
Consider  a 2D polygon or a 3D mesh.
When a polygon is closed, the sum of its edge vectors is trivially the zero vector. 
Similarly, for a closed mesh, the sum of its face 2-vectors will be the zero 2-vector. 
A direct consequence of this is that if a single edge or face is missing, its magnitude must absolutely equal  the magnitude of the sum of the remaining boundary elements. 
\begin{center}
    \includegraphics[width=0.7\linewidth]{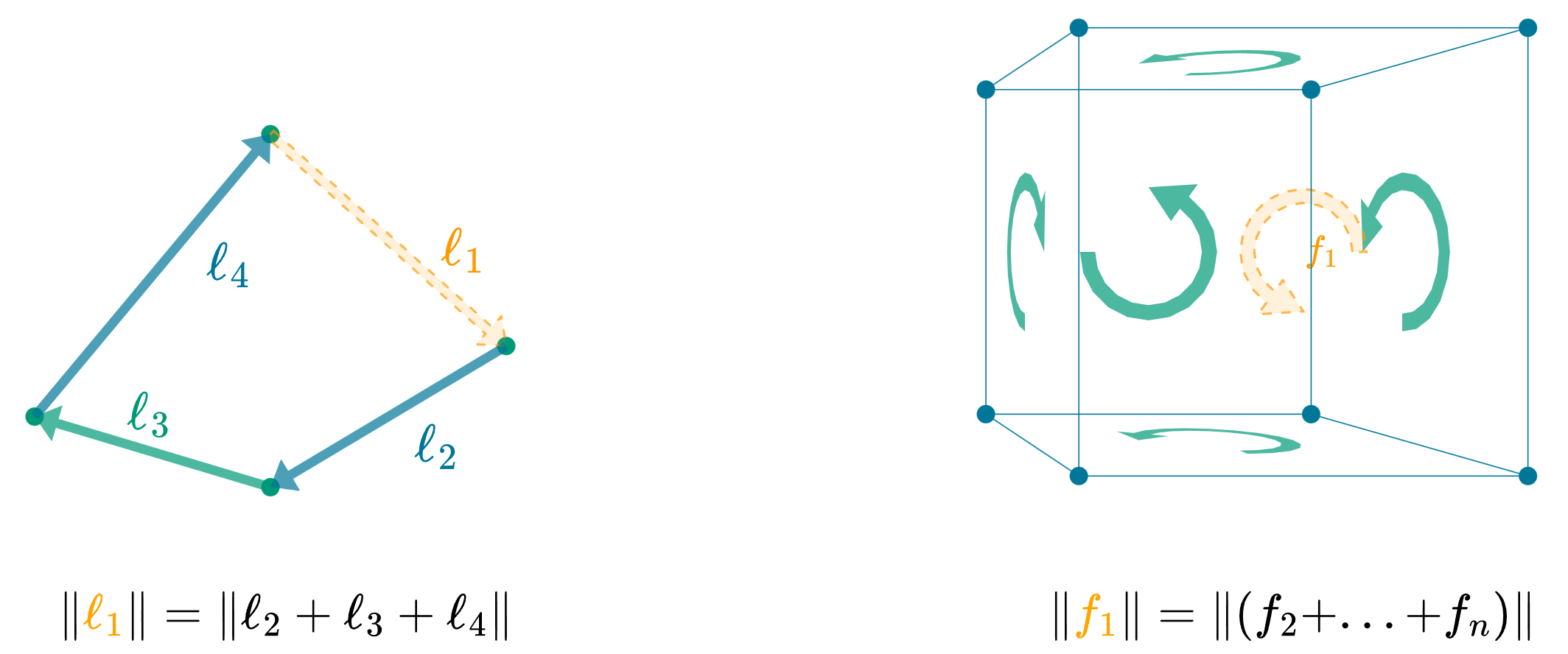}
\end{center}
We can use this idea to find an alternate form to calculate the magnitude of a $k$-simplex $\sigma_k$, by considering it the gap of a $k+1$ simplex formed by the simplicial extension of $\sigma_k$ with an arbitrary point.

Before we continue, we must introduce the boundary of a simplex, a standard concept in algebraic topology \cite{munkres1984}.
The boundary of a $k$-simplex $\sigma_k = [v_0, \ldots, v_k]$ is defined by
    \begin{equation}
        \partial \sigma_k = \sum_{i = 0}^k (-1)^i [v_0, \ldots, \xcancel{v_i}, \ldots, v_k] \ ,
    \end{equation}
where $\xcancel{v_i}$ denotes the missing term.

\begin{center}
    \includegraphics[width=0.95\linewidth]{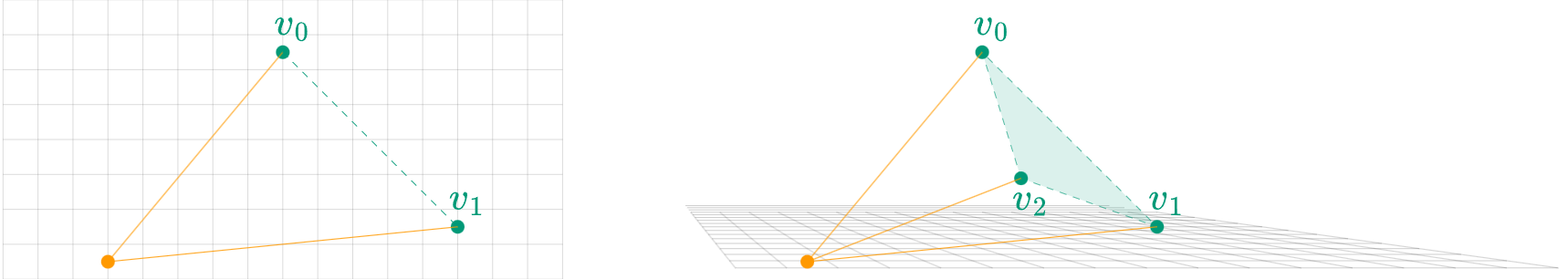}
\end{center}
The figure above (left) then suggests that we can consider any edge $[v_0, v_1]$ as the missing edge, the {\em gap}, of a triangle formed by the $v_0, v_1$ and a new point $v_2$. 
We can then use the closure of $[v_0, v_1, v_2]$ to express the gap $[v_0, v_1]$ as the sum of the remaining edges. 
Moreover, by picking the origin $o$ as the extra point $v_2$, we can factor and recognize the ideal norm $\lVert x \rVert_\infty = \lVert o \vee x \rVert $ of our boundary complex:
\begin{equation}
\lVert v_0 \vee v_1 \rVert = \lVert \overbrace{o \vee v_0 + v_1 \vee o}^\text{2 edges of triangle} \rVert 
= \lVert o \vee (v_1 - v_0) \rVert 
= \overbrace{\lVert v_1 - v_0\rVert_\infty}^\text{boundary vertices} \ .
\end{equation}
Thus, the length of the edge $[v_0, v_1]$ can be computed on the edge as $\norm{v_0 \vee v_1}$ or using its boundary $\partial [v_0, v_1] = [v_1] - [v_0]$ as $\norm{v_1 - v_0}_\infty$.
Similarly, in the right of the figure, any triangle $[v_0, v_1, v_2]$ can be viewed as the missing face of a tetrahedron formed by $v_0, v_1, v_2$ and a new point $v_3$, which we take to be the origin $\origin$:
\begin{equation}
\begin{aligned}
\lVert v_0 \vee v_1 \vee v_2 \rVert &= \lVert \overbrace{o \vee v_0 \vee v_1 + o \vee v_1 \vee v_2 + o \vee v_2 \vee v_0}^\text{three remaining faces of tetrahedron} \rVert \\
&= \lVert o \vee (v_0 \vee v_1 + v_1 \vee v_2 + v_2 \vee v_0 ) \rVert \\
&= \lVert \underbrace{ v_0 \vee v_1 + v_1 \vee v_2 + v_2 \vee v_0 }_\text{boundary edges of triangle}\rVert_\infty
\end{aligned}
\end{equation}
Hence, the area of the triangle $[v_0, v_1, v_2]$ can be computed on the face as $\norm{v_0 \vee v_1 \vee v_2}$ or using its boundary $\partial [v_0, v_1, v_2] = [v_1, v_2] - [v_0, v_2] + [v_0, v_1]$ as $\lVert { v_0 \vee v_1 + v_1 \vee v_2 + v_2 \vee v_0 } \rVert_\infty$.
These two examples clearly illustrate that the magnitude of $k$-simplex can alternatively be computed from its boundary, leading to \vref{th:magnitude}.
\begin{theorem}[Magnitude of a $k$-simplex $\sigma_k$ from the boundary $\partial \sigma_k$]\label{th:magnitude}
The magnitude (volume, area, ...) of a $k$-simplex $\sigma_k$ can be computed over the $(k-1)$-dimensional facets (areas, lengths, ...) comprising its oriented boundary $\partial \sigma_k$:
\begin{equation*}\label{magnitudeIdeal}
\begin{aligned}
\kmag \sigma_k = \frac{1}{k!} \lVert S_k \rVert &= \frac{1}{k!} \big\lVert \sum_{\sigma_{k-1} \in \partial \sigma_k} S(\sigma_{k-1}) \rVert_\infty \\
&= \frac{1}{k!} \big\lVert \sum_{i=0}^{k} (-1)^{i} \, v_0 \vee \cdots \vee \xcancel{v_i} \vee \cdots \vee v_k \rVert_\infty\ ,
\end{aligned}
\end{equation*}
where $\xcancel{v_i}$ means the vertex $v_i$ is missing from the product.
Note that the sign swaps are precisely those in the definition of the boundary $\partial \sigma_k$.
\end{theorem}
\begin{proof}
Any vertex $v_i = \origin + \vec v_i^*$, and hence
\begin{align}
    v_0 \vee \cdots \vee v_k &= ( \origin + \vec v_0^*) \vee \cdots \vee ( \origin + \vec v_k^*) \\
    &= \origin \vee \sum_{i=0}^k (-1)^i v_0 \vee \cdots \vee \xcancel{v_i} \vee \cdots \vee v_k \quad + \quad  \underbrace{\vec v_0^* \vee \cdots\vee \vec v_k^*}_\text{ideal term} \ ,
\end{align}
where we repeatedly use $\origin \vee \vec v_i^* = \origin \vee v_i$.
It therefore follows that $\norm{S_k} = \lVert \origin \vee \sum_{\partial \sigma_k} S_{k-1} \rVert = \lVert {\sum_{\partial \sigma_k} S_{k-1}} \rVert_\infty$.
\end{proof}
\noindent
The relationship between the Euclidean and ideal terms have a strong Stokes' Theorem feeling to them, as it relates relates the magnitude of a $k$-simplex to the magnitude of the facets on its boundary. 
This is not a coincidence, as will become even more clear in \cref{dynamics}.
\Cref{tab:pga-norms} gives an overview of the relationship between the Euclidean and ideal norms for various simplices.

\begin{table}[!h]
\scriptsize
\renewcommand{\arraystretch}{1.6}
\setlength{\tabcolsep}{0pt}
\centering
\rowcolors{1}{white}{tablegray}
\begin{tabular}{|@{}>{\centering\arraybackslash}p{0.19\textwidth}
                |>{\centering\arraybackslash}p{0.19\textwidth}
                |>{\centering\arraybackslash}p{0.26\textwidth}
                |>{\centering\arraybackslash}p{0.26\textwidth}@{}|}
\arrayrulecolor{tablegray}
\Xhline{0.2pt}
$\boldsymbol{x}$ & \centering\arraybackslash \textbf{simplex/carrier} & \centering\arraybackslash $\boldsymbol{\lVert x \rVert}$ & \centering\arraybackslash $\boldsymbol{\lVert x \rVert_\infty} = \lVert \origin \vee \boldsymbol{x} \rVert$ \\
\Xhline{0.02pt}
$v$ & vertex/point & 1 & length of line to origin \\
$v_0 \vee v_1$ & edge/line & length of edge & $2\times$ area of triangle with origin \\
$v_0 \vee v_1 \vee v_2$ & triangle/plane & $2\times$ triangle area & $6\times$ volume of tetra with origin \\
$v_0 \vee v_1 \vee v_2 \vee v_3$ & tetrahedron/volume & $6\times$ tetrahedron volume & $0$ \\
\arrayrulecolor{black}
\Xhline{0.2pt}
\end{tabular}
\vspace{3pt}
\caption{Norm and ideal norm for 3D PGA elements built from joining normalized vertices. The regular PGA pattern  continues into $d$ dimensions. }
\label{tab:pga-norms}
\end{table}

\subsection{A $k$-Complex and its Magnitude}\label{complexes}
\Cref{th:magnitude} holds the key to extending from simplices to simplicial complexes, enabling us to easily calculate $k$-magnitudes of $k$-complexes $\mathcal{K}_k$.
Recall that a $k$-complex $\mathcal{K}_k$ is a collection of $k$-simplices $\sigma_k$. 
For example, a triangle mesh is a collection of triangles ($2$-simplices) so it is a $2$-complex, and a 3D discrete shape can be represented as a collection of tetrahedra ($3$-simplices) in its interior so it is a $3$-complex.
The common representation os a 3D mesh is however as a $2$-complex representing its oriented, closed and manifold boundary.
That boundary representation is independent of the actual volumetric simplices composing the complex.
In practical terms, \cref{th:magnitude} enables the direct computation of areas of arbitrary planar polygons from either the volumetric or boundary representations (without triangulation), and volumes of arbitrary triangle meshes without tetrahedralization.

To illustrate, consider the polygon shown in the figure below.
It can be represented using either a $2$-chain $\mathcal{K}_2 = [v_0, v_1, v_2] + [v_0, v_2, v_3]$, or using its $1$-chain boundary $\mathcal{K}_1 = \partial \mathcal{K}_2 = [v_0, v_1] + [v_1, v_2] + [v_2, v_3] + [v_3, v_0]$.
\begin{center}
            \includegraphics[width=0.7\linewidth]{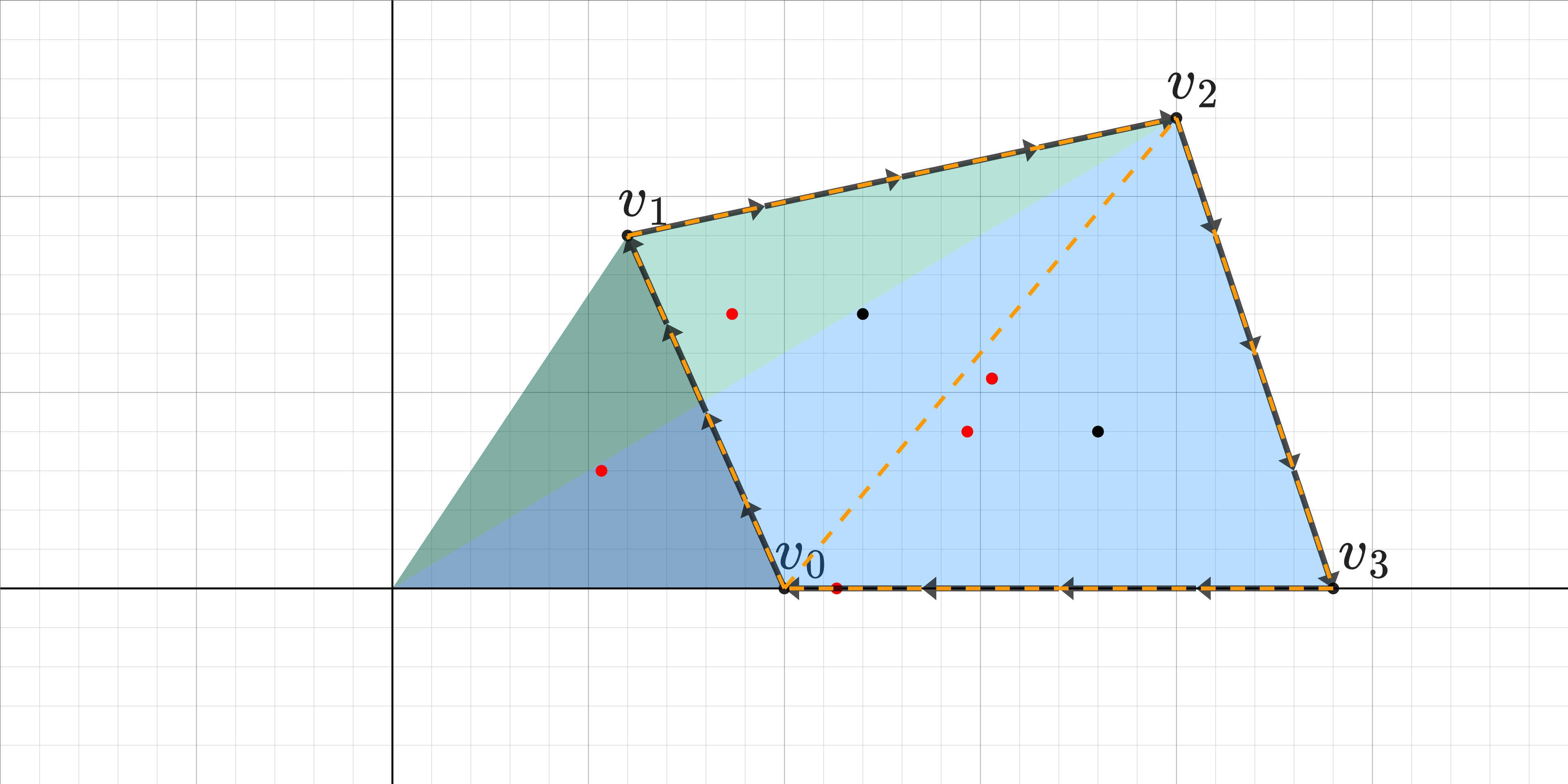}
\end{center}
Extending the lessons of \cref{sec:mindthegap} from simplices to a complexes, the signed area $A$ of the polygon can be computed using either $\mathcal{K}_2$, or from the boundary $\partial \mathcal{K}_2$ by joining in an arbitrary apex point (taken to be the origin $\origin$):
    \begin{align}
        A &= \tfrac{1}{2} \sum_{\sigma_2 \in \mathcal{K}_2} S(\sigma_2) = \tfrac{1}{2} ( v_0\vee v_1 \vee v_2 + v_0 \vee v_2 \vee v_3) \label{eq:area_from_surface} \\
        &= \tfrac{1}{2} \ \origin \vee \sum_{\sigma_1 \in \partial \mathcal{K}_2} S(\sigma_1) = \tfrac{1}{2} \, \origin \vee ( v_0\vee v_1 + v_1 \vee v_2 + v_2 \vee v_3 + v_3 \vee v_0) \label{eq:area_from_boundary}
            \end{align}
Note that if the origin $\origin$ (which can be freely chosen) coincides with any of the $v_i$, \cref{eq:area_from_boundary} reduces to \cref{eq:area_from_surface}, which shows that the two are indeed equal.
But if only the unsigned area is required, \cref{eq:area_from_boundary} simplifies to $\lVert \sum_{\sigma_1 \in \partial \mathcal{K}_2} S(\sigma_1) \rVert_\infty$, see \cref{th:magnitude}. Note that to correctly compute the area it is important to first sum and then take the ideal norm, because we need to use signed areas in order for overlapping regions to cancel one another.
Hence the (unsigned) $k$-magnitude of a $k$-complex $\mathcal{K}_k$ is given by the following theorem.
\begin{theorem}[Magnitude of a $k$-complex $\mathcal{K}_k$]\label{th:magnitudeC}
The $k$-dimensional magnitude of a $k$-complex $\mathcal{K}_k$ is 
a function of the $k$-simplices in its interior, or 
of the $(k-1)$-simplices on its oriented boundary $\partial \mathcal{K}_k$:
\begin{equation}\label{hyperVolume}
\kmag \mathcal{K}_k \coloneqq \sum\limits_{\mathcal{K}_k} \kmag \sigma_k 
= \frac{1}{k!} \sum\limits_{\mathcal{K}_k} \lVert S(\sigma_k) \rVert
= \frac 1 {k!}\lVert \sum\limits_{\partial {\mathcal{K}_k}} S(\sigma_{k-1})\rVert_\infty 
 \end{equation}
\end{theorem}
Before we provide the proof to this theorem, it is important to stress that the sum runs directly over the boundary of the $k$-complex, and not over the boundary of every individual $k$-simplex in the complex. 
This is due to the fact that contributions from internal boundaries of $k$-simplexes in the complex cancel, as the following figure shows, where the $1$-chains $[v_0, v_2] = -[v_2, v_0]$, a property reflected in the anti-symmetry of the join product $v_0 \vee v_2 = - v_2 \vee v_0$.
\begin{center}
    \includegraphics[width=0.4\textwidth]{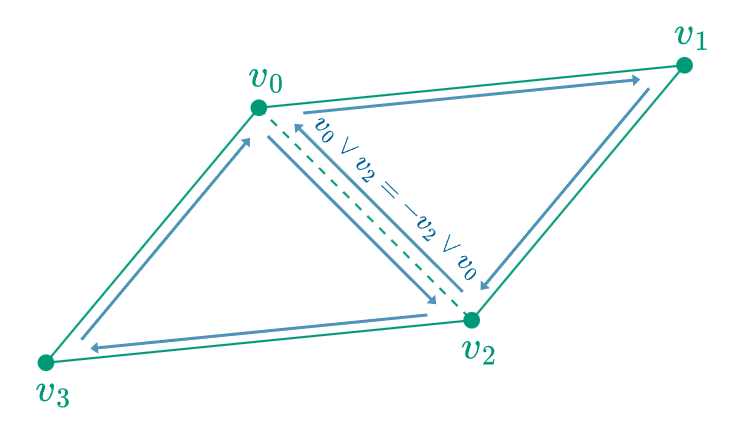}
\end{center}
We now proceed with the proof to \cref{th:magnitudeC}.
\begin{proof}
Recall from the proof of \cref{th:magnitude} that $ S(\sigma_k) = \origin \vee \sum_{\partial \sigma_k} S(\sigma_{k-1}) + \text{ideal}$, so by direct computation we find
\[ \sum\limits_{\sigma_{k} \in \mathcal{K}_k} S(\sigma_k)
= \origin \vee \sum\limits_{\sigma_{k} \in \mathcal{K}_k} \, \sum\limits_{\sigma_{k-1} \in \partial \sigma_k} S(\sigma_{k-1}) + \text{ideal term} \ ,  \]
where we have been verbose to clarify every sum.
However, this sum includes all internal simplices $\sigma_{k-1}$, which cancel due to the anti-symmetry of the join product if the $k$-complex $\mathcal{K}_k$ is consistently oriented, leaving only the boundary vertices to contribute to the sum, and thus
    \[ \sum\limits_{\mathcal{K}_k} \sum\limits_{\partial \sigma_k} S_{k-1} = \sum\limits_{\partial \mathcal{K}_k} S_{k-1} \ . \]
Hence, $\lVert \sum\limits_{\mathcal{K}_k} S_k \rVert = \lVert \origin \vee \sum\limits_{\partial \mathcal{K}_k}  S_{k-1} \rVert = \lVert \sum\limits_{\partial \mathcal{K}_k} S_{k-1} \rVert_\infty$.
\end{proof}

In conclusion, we can compute the magnitude of a complex $\mathcal{K}_k$, which was originally defined in terms of a simplex decomposition, entirely in terms of the sum of the boundary simplices $\sum_{\partial \mathcal{K}_k} S_{k-1}$, by taking the ideal norm of this sum.
We will now investigate what the significance of the Euclidean norm of this sum over the boundary simplices is.
The sum is zero if the original complex is closed; but if the complex is not, 
then the Euclidean norm of this sum is actually the norm of the sum of the $(k-1)$-dimensional simplex facets missing from the boundary. 
This leads to the following corollary:

\begin{corollary}[magnitude of complex gap]\label{col:gap}
Given a $k$-complex $\mathcal{K}_k$ with a non-closed boundary $\partial \mathcal{K}_k$, i.e. $\partial \partial \mathcal{K}_k \neq 0$, 
the $(k-1)$-magnitude of the complex gap $\xcancel{\partial} \mathcal{K}_k$ is identical to that of $\partial \mathcal{K}_k$, and hence given by
\begin{equation}\label{gap}
\kmag \xcancel{\partial} \mathcal{K}_k
=\kmag \partial \mathcal{K}_k
= \frac {1} {(k-1)!} \lVert \sum\limits_{\partial \mathcal{K}_k} { S_{k-1} }\lVert
\end{equation}
\end{corollary}
\noindent
Hence, the magnitude of a complex gap can be measured directly, without needing to explicitly reconstruct this gap first.
Note again that it is the norm of the sum, and not the sum of the norm. 
This magnitude will equal exactly the area of the missing boundary simplices iff those simplices share the same 
carrier. (i.e., if they are edges on the same line or triangles in the same plane, etc.).
We shall now give a hands-on example that demonstrates the power of \cref{col:gap}.

\subsection{Example: Volume of a Capped Mesh}\label{example 1}
Let us use the formulas from \vref{table:cheat}, to solve a real world 3D engineering problem. Given a closed manifold mesh of an airplane fuel tank, and a plane representing the fuel's top surface, determine the remaining amount of fuel, as well as track the evolving center of mass\footnote{This problem was first posed via bivector.net by Dr Todd Ell of Collins Aerospace. The solution presented is in use today.}. 
The setup is shown in the figure on the first page.

After embedding our vertices, edges and faces as their respective PGA carriers using the unnormalized join of  1, 2 and 3 points respectively, our solution has two parts. First we determine which triangles are under the fuel plane, subdividing triangles that cross the level plane. Next, we calculate the volume and center of mass of the, no longer closed, mesh formed by the triangles under the plane.

PGA's tools will turn out to be well suited to handle both tasks. Determining which triangles are under the plane requires finding the side of the fuel plane a given vertex is on. The PGA join $\vee$ provides the perfect tool for this, as we can directly produce a signed volume by joining the fuel plane with the tested point. Worked out in coefficients, $p \vee v_i$ produces the familiar efficient expression: 

\begin{align}
\label{test:side}
p \vee v_i &= (a\mathbf e_{1} + b\mathbf e_{2} + c\mathbf e_{3} + d\mathbf e_{0}) \vee (x\mathbf e_{032} + y\mathbf e_{013} + z\mathbf e_{021} + \mathbf e_{123})  \\ &= ax+by+cz+d. \nonumber
\end{align}

Armed with \vref{test:side}, we iterate over all triangles in our mesh, discarding those that are entirely above the plane based on the sign of $p \vee v_i$, while retaining those entirely below. The triangles we retain are not stored, instead their associated (and precalculated) carrier plane is accumulated in a running sum $F = \sum v_0 \vee v_1 \vee v_2$. The triangles that intersect the plane must be split into three new triangles. Here the PGA meet $\wedge$ operation is the perfect tool. The vertex created when an edge $E$ intersects a plane $p$ is found as
\begin{equation}
\label{intersect}
\begin{aligned}
E \wedge p = &\,(d_x\mathbf e_{01} + d_y\mathbf e_{02} + d_z\mathbf e_{03} + z\mathbf e_{12} + y\mathbf e_{31} + x\mathbf e_{23}) \wedge (a\mathbf e_{1} + b\mathbf e_{2} + c\mathbf e_{3} + d\mathbf e_{0}) \\
= &\, (bd_z-cd_y-dx)\mathbf e_{032} + (cd_x-ad_z-dy)\mathbf e_{013} +\\ &\, (ad_y-bd_x-dz)\mathbf e_{021} + (ax+by+cz)\mathbf e_{123}.
\end{aligned}
\end{equation}
For triangles that have a single vertex under the fuel plane, we construct and normalize the new vertices, join them with the vertex under the plane into a new triangle and add the resulting vector to the running sum.
For triangles with a single vertex above the plane, we can avoid having to construct the two new triangles under the plane by realizing that all three subdivided triangles must sum up to our (precalculated) starting triangle. So in this scenario we can construct the single triangle above the plane and subtract it from our precalculated value for the entire triangle.
\noindent\makebox[\textwidth][c]{\includegraphics[width=1.0\textwidth]{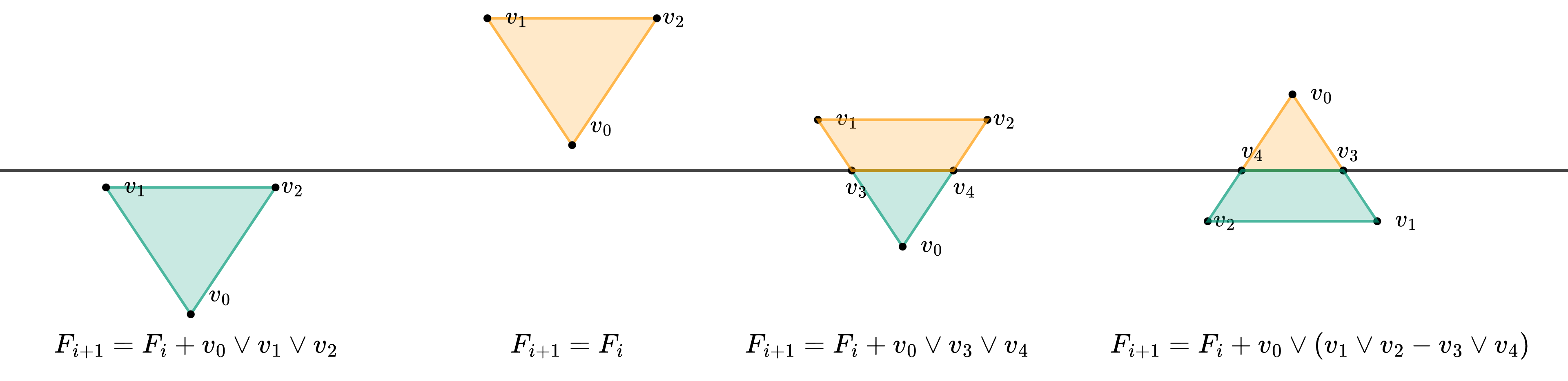}
}
At the end of this iteration we have a sum $F$, a vector that is simply the sum of all triangles under the intersection plane. As we saw before, this is all that is needed to calculate the volume under the plane, even though this mesh is not closed.
By starting with a manifold closed mesh, we know that all missing triangles lie in the same carrier, namely the intersection plane itself. Because of that, $a = \tfrac{1}{2}\lVert F \rVert$ is exactly the area of the gap. We also know the signed distance of the intersection plane $h =  \overline{p} \vee \origin$ to the origin, hence we can calculate the signed volume of the simplicial extension of the gap as $\tfrac{1}{3}ah = \tfrac{1}{6}\lVert F \rVert(\overline p \vee \origin) $ enabling us to express the total volume as
\begin{equation}
\label{fuelVolumeOld}
\begin{aligned}
V_{\text{fuel}} 
&= \frac 1 6 \big (F \vee \origin +  \lVert F\rVert (\overline p \vee \origin)\big) \\
\end{aligned}
\end{equation}
Understanding the arbitrary nature of the origin in this equation allows us to simplify further as for any point $o' \in p$, we have $\overline p \vee o' = 0$ so the second term vanishes. Hence we pick a point $o' = (o \cdot p)p^{-1}$ and find
\begin{equation}
\label{fuelVolume}
V_\text{fuel} = \frac{1}{6}  F \vee o'
\end{equation}
Recalling \cref{test:side} we note that if $o'$ is indeed chosen as  the origin ($x=y=z=0$), the expression $F \vee o'$ for a vector $F$ is simply the extraction of the $\mathbf e_0$ or $d$ coefficient of its associated linear equation $ax + by + cz + d = 0$. The other coefficients of $F$ are then unused, and do not need to be accumulated or calculated for the intersecting triangles, reducing the accumulator to a sum of these $d$ values. For the full triangles they are precalculated and for a new intersecting triangle, in function of the coordinates $x_i, y_i, z_i$ of its three points works out to just
\[
d=x_3(y_2z_1-y_1z_2)+x_2(y_1z_3-y_3z_1)+x_1(y_3z_2-y_2z_3)
\]
When the intersection plane does not pass through the origin, we can still pick an $o'$ with two zero coefficients, by intersecting one of the coordinate axis with $p$. In this case only two coefficients of $F$ need to be calculated and accumulated, and from \cref{test:side} we note that the join with $o'$ now involves only a single multiply-and-add.

Finally, we can halve the average runtime once more by noting that the computational cost is dominated by the number of triangles included in the sum. Since the total volume of the fuel tank is fixed, we don’t always need to integrate the part below the plane. When the tank is nearly full, it is more efficient to calculate the (smaller) volume above the plane and subtract that from the total volume.

\section{Moments of Simplices and Complexes}\label{dynamics}
Now that we have seen how to compute $k$-magnitudes for simplices and complexes of arbitrary $k$, we turn our attention to the computation of dynamical quantities such as center of mass and inertia.
Before we consider these specific quantities however, we briefly discuss arbitrary functions over a $k$-complex to establish some basic concepts and notation.
\subsection{Integrating Multivector-valued Functions over a Complex}

Let $f(\sigma_k)$ be a multivector-valued function defined on a complex $\mathcal{K}_k$. How the function processes the input $\sigma_k$ is up to the function, e.g. $f(\sigma_k) = \tfrac{1}{k+1}\sum_{i=0}^k v_i$ would return the center of mass of $\sigma_k$.
Then the integral of $f(\sigma_k)$ over a complex $\mathcal{K}_k$ is given by
\begin{equation}\label{eq:integral}
    F(\mathcal{K}_k) = \sum_{\sigma_k \in \mathcal{K}_k} f(\sigma_k) S({\sigma_k}) \ ,
\end{equation}
where $S({\sigma_k})$ ensures that $f(\sigma_k)$ is weighted by the signed $k$-magnitude of $\sigma_k$.
Using the same technique as in \cref{th:magnitude,th:magnitudeC} the integration can also be performed over the boundary of the $k$-complex $\mathcal{K}_k$:
\begin{align}
    F(\mathcal{K}_k) &= \sum_{\sigma_k \in \mathcal{K}_k} f(\sigma_k) S({\sigma_k}) = \sum_{\sigma_{k-1} \in \partial \mathcal{K}_k} f([\origin,\sigma_{k-1}]) (\origin \vee S({\sigma_{k-1}}))  +
    \text{ideal term} \ .
    \end{align}
We have already seen that when $f(\sigma_k) = 1$ we obtain the signed $k$-magnitude, and in the following sections we will see how the integral of $f(\sigma_k) = \tfrac{1}{k+1}\sum_{i=0}^k v_i$ leads to the center of mass (c.o.m.), and there is even a function that leads to the inertia tensor.

\subsection{Center of Mass of a Mesh/Complex}

\begin{theorem}[Center of mass]
In $k$-D PGA,  a $k$-simplicial complex $\mathcal{K}_k$, defined by
a boundary of $(k-1)$-simplices $\partial \mathcal{K}_k$, with a  uniform mass distribution,
has a center of mass $C_\text{com}$  given by
\begin{equation}\label{com}
\begin{aligned}
C_\text{com} = \frac 1 {\lVert C \rVert} \int_C \vec x dV &= \frac 1 {\lVert C \rVert (k+1)!} \sum\limits_{\partial \mathcal{K}_k} ( \sum_{i=0}^{k-1} v_i + \origin )  (S_{k-1} \vee  \origin)\\
&= \frac 1 {\lVert C \rVert (k+1)!} \sum\limits_{\partial \mathcal{K}_k} ( v_0 + \ldots + v_{k-1} + \origin )(v_0 \vee \cdots \vee v_{k-1} \vee \origin)  \\
\end{aligned}
\end{equation}
\end{theorem}
Where in practice, we prefer to leave out the factor $\tfrac{1}{\lVert C \rVert}$, and arrive with a homogeneous point at the center of mass scaled with the volume. In this form several such quantities can be added to find the composite c.o.m.

We illustrate this in \vref{fig:com}, where the choice of extra point produces the same area and center of mass for our pentagon.

\begin{figure}
    \centering
    \includegraphics[width=1.0\linewidth]{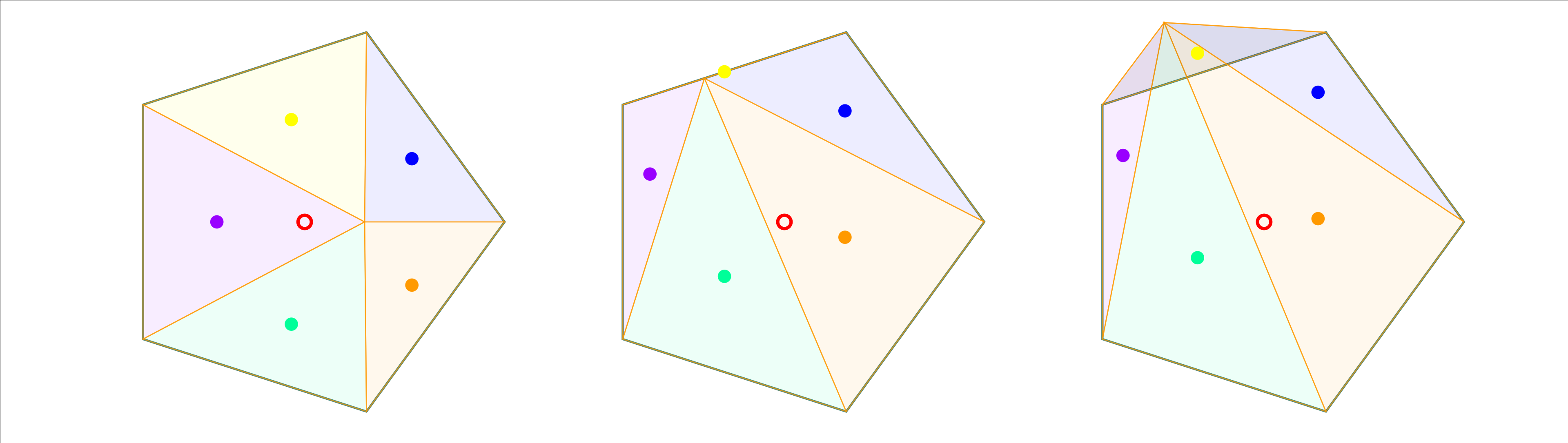}
  \caption{ The area of the pentagon is the sum of the signed areas of the triangulation w.r.t. to any point. The c.o.m. of the pentagon, drawn as red circle, is the area-weighted-average of the c.o.m. of
  the individual triangles, drawn as solid dots. In $d$-D, the triangles become pyramidal cones from the arbitrary apex point to the boundary facets, still contributing their hypervolume-weighted centroids to the centroid of the complex. But the polygon centroids are then $(d-k)$-blades, whose intersection with the $k$-blade representing the carrier plane gives a centroid point $d$-blade.}
    \label{fig:com}
\end{figure}

\subsection{Example Continued}

We continue the example above, to compute the center of mass of the sliced fuel tank. 
We can again proceed without knowing the exact location of the missing triangles by picking our $o' \in p$. We now find

\begin{equation}
    F_\text{com} = \frac{1}{24} \sum (v_0 + v_1 + v_2 + o')(v_0 \vee v_1 \vee v_2 \vee o') .
\end{equation}
The resulting homogeneous point is then located at the center of mass, and scaled with the total volume. As before a proper choice of $o'$ minimizes the computations needed, and as before the homogeneous scale ensures that the sum of the center of mass below the plane and that above equals the total center of mass. 

The total volume of the fuel can be computed from this center of mass as its Euclidean norm, hence no separate computation is required.
An interactive demonstration is available \cite{com-demo}.

\subsection{Inertia of a Mesh/Complex}

In an earlier article on kinematics and dynamics in PGA \cite{dorst-2024} we assumed an arbitrary rigid body was aligned to its principal axis and that the principal moments were known. Using those, one can then construct an inertial duality map. 
Finding that inertial frame for an arbitrary mesh was left as an exercise to the reader. Now that we have calculated the 0th moment (volume) and 1st moment (center of mass) for an arbitrary mesh, it is a natural question if we can do the same for the 2nd moment - the inertia tensor.

While the scalar volume and vector center of mass have natural representations in PGA, the inertia tensor is a rank 2 tensor that has no direct PGA equivalent. We can however still use the same techniques once we realize that this rank 2 tensor can be represented as a set of three vectors: the same vectors that would make up its matrix representation.

Using this representation we can then accumulate inertia over the boundary of our mesh as before, and directly construct the $2k$-reflection (aka {\em rotor}) required for diagonalization, enabling us to align our arbitrary mesh with its inertial frame.
This allows us to avoid using matrix representations where we would need
to convert the aligning matrix back to a motor. 
\begin{definition}[Tetrahedron Inertial Frame]
The canonical inertia around the origin for a solid tetrahedron defined by three 
points at positions $[x_i, y_i, z_i]$ and the origin can
be described by a frame of three vectors $I_1,I_2,I_3$, converted from a scalar version in the literature (e.g. \cite{tonon-2005}), can be written as
\begin{equation}\label{inertia}
\begin{gathered}
I_1 = (I_Y + I_Z) \mathbf e_1 + I_{XY} \mathbf e_2 + I_{XZ} \mathbf e_3 \\
I_2 = I_{XY} \mathbf e_1 + (I_Z + I_X)\mathbf e_2 + I_{YZ} \mathbf e_3 \\ 
I_3 = I_{XZ} \mathbf e_1 + I_{YZ} \mathbf e_2 + (I_X + I_Y)\mathbf e_3 \\
I_{X} = X\cdot(X + X')\\
I_{XY} = -X\cdot Y - \frac 1 2 (X\cdot Y' + Y\cdot X') \\
X =  x_1 \mathbf e_1 + x_2 \mathbf e_2 + x_3 \mathbf e_3, \qquad X' = x_2 \mathbf e_1 + x_3 \mathbf e_2 + x_1 \mathbf e_3 
\end{gathered}
\end{equation}
\end{definition}
\noindent
with the definitions of $\{Y, Z\}$, $\{Y', Z'\}$, $\{I_{Y}, I_{Z}\}$ and $\{I_{XZ}, I_{YZ}\}$ following those for $X$, $X'$, $I_{X}$ and $I_{XY}$ respectively. 
This inertial frame $I_i$ can now be weighted with the signed volume and accumulated for all triangles in the mesh - similar to our calculation of the center of mass.
\begin{equation}\label{meshInertia}
{I_\text{tot}}_i = 
\frac{1}{10\lVert C \rVert}\sum\limits_{\partial C} I_i(\origin \vee v_0 \vee \cdots \vee v_{k-1}) . 
\end{equation}
Such a set of vectors $I_\text{tot}$ (corresponding to a symmetric matrix) can always be diagonalized, which we will do in our framework using the Jacobi algorithm. For this we first need to define the GA equivalent of a matrix similarity transformation ($O A O^\top$) by an orthogonal matrix ($O$); in our case a {\em rotor} (normalized $2k$-reflection) applied to such a frame.
\begin{definition}[Similarity Transformation of a frame $I_i$]
Given a rotor $R$, and
a frame of vectors $I_i$, the similarity transformation which in matrix language would be written
$\mathbf{I' = RIR^{-1}}$ can be calculated in PGA as
\end{definition}
 \begin{equation}\label{similarity}
I_i' =  \sum_{j=1}^{3} R (\mathbf e_j\cdot(\tilde R \mathbf e_j R)I_j) \tilde R
  \end{equation}
This operation plays a central role in the Jacobi algorithm that we will use to
diagonalize this frame, and find the eigenframe of a simplicial complex.
\begin{definition}[Diagonalized Eigenframe]
Given a set of three symmetric vectors $I_i$, there
exists a rotor $R$ so that $I_i'$, the similarity transformation of $I$ by $R$, is a set of
eigenvalue scaled basis vectors.
\end{definition}
 \begin{equation}\label{diagonal}
I_i' =  \sum_{j=1}^{3} R (\mathbf e_j\cdot(\tilde R \mathbf e_j R)I_j) \tilde R = \lambda_i \mathbf e_i
  \end{equation}
To find the $R = R_1R_2 \cdots R_k$, we  perform a set of consecutive Givens rotations, following the
standard implementation of the Jacobi eigenvalue algorithm. We iterate over the principal
bivectors $\mathbf e_{pq}$, generating a rotor $R_k$ in the $k$-th iteration step as
\begin{equation}\label{jacobi}
\begin{aligned}
R_k &= e^{{ -\frac 1 2 \phi \mathbf e_{pq}}} \\
\phi &= \tan^{-1} \big(  \cfrac {\text{sign } \tau} {\lvert \tau \rvert + \sqrt{1 + \tau^2}}  \big) \\
\tau &= \cfrac {I_q\cdot \mathbf e_q-I_p \cdot \mathbf e_p} {2 I_p\cdot \mathbf e_q}
\end{aligned}
\end{equation}
Updating our frame $I_i$ each step with $R_k$ using \vref{similarity}. Steps where
the denominator of $\tau$ are very small are skipped as the $I_p$ vector is sufficiently
orthogonal to $\mathbf e_q$. The algorithm terminates when this denominator is sufficiently
small for all basis planes.
The reverse of the resulting rotor $R$ can then be used to align a mesh with its
eigenframe, while the resulting frame satisfies $I_i' = \lambda_i \mathbf e_i$. The
matching eigenvectors are $R\mathbf e_i\tilde R$.
This technique avoids the otherwise needed conversion from matrices to motors, and allows us to calculate the inertial bivector used in \cite{dorst-2024}, within the PGA framework.

\section{Conclusion}

The join operation in 3D PGA uses the vertices of a mesh to build an algebraic representation that compactly contains the relevant geometric information to compute various magnitudes: area, volume, center of mass, and inertia.
This construction and the subsequent formulas are coordinate-free,  equivariant under Euclidean transformations, and can be computed efficiently. The Euclidean split was our tool in developing these equations, and the resulting formulas  took the form of `norm of sums' or `sum of norms' of the Euclidean or Ideal terms of the geometrical elements. 

PGA is also very well suited to compute the physical motion of objects, and the  computations of center of mass and inertia of a mesh-given object displayed in this paper tie in perfectly with the PGA-bivector-based treatment of Newtonian mechanics in \cite{dorst-2024}. That framework unifies the translational and rotational motion equations in a manner that makes them integrable, by actually purposely ignoring the Euclidean split. We recommend it as your next read. 

In this paper, we have specialized to 3D PGA $\R_{3,0,1}$. The dimensional recursion we employed in the treatment of meshes can actually be generalized to arbitrary dimensions, and the formulas can be written in a manner that makes them dimension-agnostic (by using $\eo^*$ rather than $\origin$ for the origin, which is then automatically the $d$-dimensional Euclidean pseudoscalar). Coding in $\R_{d,0,1}$, we can thus use exactly the same 3D software to treat 2D polygon meshes and their properties, merely setting $d=2$ in the first line of code. The algebra always generates the required entities allowed in the required dimension for each operation, no more and no less. No other framework can do this.

The present paper is a harbinger; you may have noticed glimpses of the connection to Stokes' theorem in this paper which need to be developed in more detail.  We intend to produce a series of papers spelling out the advantages of PGA for Euclidean geometry and Newtonian physics, tying it to applications in CGI and robotics, with accompanying tutorials on \href{http://bivector.net}{bivector.net}. We are convinced that PGA $\R_{d,0,1}$ is the preferred framework in those fields, encompassing all that went before, in any dimensionality, in a structural and efficient manner.

\bibliographystyle{splncs04}
\bibliography{biblio}
\end{document}

%% file: definitions.tex
\usepackage{physics}
\usepackage{mathtools}  
\newcommand{\R}{\mathbb{R}}
\newcommand{\GR}[2][{}]{\ensuremath{\mathbb{R}^{{#1}}_{#2}}}

            \newcommand{\e}{\mathbf {e}_}            \newcommand{\eo}{\e{0}}

\newcommand{\Eu}[1]{\ensuremath{\text{E}({#1})}}
\newcommand{\SE}[1]{\ensuremath{\text{SE}({#1})}}

\newcommand{\gp}{\,}

\newcommand{\inv}[1]{{#1}^{-1}}
\newcommand{\half}{{\mbox{$\frac{1}{2}$}}}
\newcommand{\edg}{E}  \newcommand{\ideal}{I} \newcommand{\origin}{o}

\DeclareMathOperator{\kmag}{mag}